\documentclass[DIV12,11pt]{scrartcl}

\usepackage{microtype}

\usepackage{amsthm,amssymb,amsmath}
\usepackage[noblocks]{authblk}
\usepackage{paralist}

\title{New Abilities and Limitations of \\Spectral Graph Bisection}
\author[1]{Martin  R. Schuster}
\author[1]{Maciej Li\'{s}kiewicz}
\affil[1]{Institute of Theoretical Computer Science, University of L\"{u}beck, Germany}

\theoremstyle{plain}
\newtheorem{theorem}{Theorem}[section]
\newtheorem{lemma}[theorem]{Lemma}
\newtheorem{corollary}[theorem]{Corollary}

\newtheorem{proposition}[theorem]{Proposition}
\newtheorem{claim}[theorem]{Claim}

\newcommand\eqcomment[1]{&& \text{\scriptsize\textcolor{black!70}{\vrule height 14pt depth 7pt width 0.4pt \ $#1$}}}
\newcommand\eqlabel[1]{\stepcounter{equation}\tag{\theequation}\label{#1}}

\usepackage{tikz}
\usetikzlibrary{graphs,datavisualization,datavisualization.formats.functions}

\usepackage{paralist}
\usepackage{wrapfig}
\usepackage[ruled,vlined,english,linesnumbered]{algorithm2e}

\usepackage[utf8]{luainputenc} 
\usepackage{listings}
\usepackage{floatflt}
\usepackage{wrapfig}
\usepackage{csvsimple}
\usepackage{hyperref}

\newcommand{\caO}{{\cal O}}
\newcommand{\myvec}[1]{{\mathbf #1}}

\newcommand{\swozero}{S\setminus \{\myvec{0}\}}
\newcommand{\rwoconst}{\mathbb{R}^n_{\ne c\myvec{1}}}

\def\bmat#1,#2,#3,#4\emat{\begin{pmatrix}#1&#2\\#3&#4\end{pmatrix}}

\newcommand{\bw}{\operatorname{bw}}
\newcommand{\opt}{\mathrm{opt}}
\newcommand{\diag}{\operatorname{diag}}
\newcommand{\tsum}{\operatorname{sum}}
\newcommand{\cutwidth}{\operatorname{cw}}

\newcommand{\tr}{\operatorname{tr}}

\newcommand{\cG}{\mathcal{G}}
\newcommand{\cR}{\mathcal{R}}
\newcommand{\cO}{\mathcal{O}}

\newcommand{\cA}{\mathcal{A}}

\allowdisplaybreaks

\makeatletter
\g@addto@macro \normalsize {%
 \setlength\abovedisplayskip{4pt}
 \setlength\belowdisplayskip{4pt}
}
\makeatother

\begin{document}
\date{}

\maketitle

\begin{abstract}
Spectral based heuristics belong to well-known commonly used methods 
which determines provably minimal graph bisection or outputs ``fail'' when 
the optimality cannot be certified. In this paper we focus on Boppana's 
algorithm which belongs to one of the most prominent methods of this type.
It is well known that the algorithm works well in the random 
\emph{planted bisection model} -- the standard class of graphs for analysis
minimum bisection and relevant problems. In 2001 Feige and Kilian posed the question 
if Boppana's algorithm works well in the semirandom model by Blum and Spencer.
In our paper we answer this question affirmatively. We show also that the 
algorithm achieves similar performance on graph classes which extend  
the semirandom model. 

Since the behavior of Boppana's algorithm on the semirandom 
graphs remained unknown, Feige and Kilian proposed a new semidefinite programming 
(SDP) based approach and proved that it works on this model. The relationship between 
the performance of the SDP based algorithm and Boppana's approach was left 
as an open problem. In this paper we solve the problem in a complete way
by proving that the bisection algorithm of Feige and Kilian provides exactly 
the same results as Boppana's algorithm.
As a consequence we get that Boppana's algorithm achieves the optimal 
threshold for exact cluster recovery in the \emph{stochastic block model}. 
On the other hand we prove some limitations of Boppana's approach:
we show that if the density difference on the parameters of the 
planted bisection model is too small then the algorithm fails with high 
probability in the model.
\end{abstract}

\section{Introduction}

The minimum graph bisection problem is one of the classical
NP-hard problems \cite{garey1976some}: for an undirected graph $G$
the aim is to partition the set of vertices $V=\{1,\ldots, n\}$ ($n$ even) 
into two equal sized sets, such that the number of cut edges, i.e.  edges with  
endpoints in different bisection sides, is minimized. 
The bisection width of a graph $G$, denoted by $\bw(G)$, 
is then the minimum number of cut edges in a bisection of $G$.
Due to practical  significance in VLSI design, image processing, computer vision 
and many other applications
(see \cite{lengauer2012combinatorial, bhatt1984framework,
wu1993optimal, kwatra2003graphcut, lipton1980applications, schloegel2000graph})
and its theoretical importance, the problem has been the subject of a considerable amount  
of research from different perspectives: approximability
\cite{saran1995finding,arora1995polynomial,feige2000approximating,Feige2002,khot2006ruling},
average-case complexity \cite{Bui1987}, and 
parameterized algorithms \cite{marx2006parameterized,van2013parameterized}
including the seminal paper in this field  by  Cygan et~al.  \cite{Cygan2014} 
showing that the minimum bisection is fixed parameter tractable.

In this paper we consider polynomial-time algorithms
that for an input graph either output the  \emph{provable} minimum-size 
bisection or ``fail''  when the optimality cannot be certified.
The methods should work well for all (or almost all, depending on the model) 
graphs of particular classes,
i.e.~provide for them a certified optimum bisection, 
while for irregular, worst case instances the output 
can be ``fail'', what is justifiable.
We investigate two well-studied graph models: 
the \emph{planted bisection model} and its extension the \emph{semirandom model} 
which are widely used to analyze and benchmark graph partitioning algorithms. We refer to  
\cite{Bui1987,dyer1989solution,Boppana1987,blum1995coloring,condon2001algorithms,Feige2001,carson2001hill,mcsherry2001spectral,bollobas2004max,Coja2005,makarychev2012approximation} to cite some of the relevant works.
Moreover, we consider the 
\emph{regular graph model} introduced of Bui et al.~\cite{Bui1987}
and a new extension of the semirandom model. 
For a (semi)random model
we say that some property 
is satisfied with high probability (w.h.p.) if the probability that
the property holds tends to $1$ as the number of vertices $n\to \infty$. 

In the planted bisection model, denoted as $\cG_{n}(p,q)$
with parameters $1>p=p(n) \ge q(n) = q > 0$, 
the vertex set $V=\{1,\ldots, n\}$ is partitioned randomly into 
two equal sized sets  $V_1$ and $V_2$, called the \emph{planted bisection}. 
Then for every pair of vertices do independently: 
if both vertices belong to the same part of the bisection 
(either both belong to $V_1$ or both belong to $V_2$) then include an edge 
between them with probability $p$; If the two vertices belong to different parts,
then connect the vertices by an edge with probability $q$. 
In the semirandom model for graph bisection \cite{Feige2001}, 
initially a graph $G$ is chosen at random according
to model $\cG_{n}(p,q)$. Then a monotone adversary is allowed to modify
$G$ by applying an arbitrary sequence of the following monotone transformations:
(1) The adversary may remove from the graph any edge
crossing a minimum bisection;
(2) The adversary may add to the graph any edge not crossing
the bisection.
Finally, in the regular random model, denoted as  $\cR_{n}(r,b)$,
with $r=r(n) < n$ and $b=b(n) \le (n/2)^2$,
the probability distribution is uniform on the set of all graphs
on $V$ that are $r$-regular and have bisection width $b$. 
 
The planted bisection model was first proposed in the sociology literature 
\cite{holland1983stochastic} under the name \emph{stochastic block model}
to study community detection problems in random graphs. In this setting, the 
planted bisection $V_1,V_2$ (as described above) models latent communities 
in a network and the goal here is to recover the communities from the observed graph.
In the general case, the model allows some errors by recovering, multiple communities, 
and also that $p(n) < q (n)$. The community detection problem on the stochastic block model
has been subject of a considerable amount of research in physics, statistics and computer science
(see e.g. \cite{abbe2017community,moore2017computer} for current surveys).
In particular, an intensive study has been carried out on providing lower bounds 
on $|p-q|$ to ensure recoverability of the planted bisection.  

The main focus of our work is the bisection algorithm proposed by Boppana \cite{Boppana1987}.
Though introduced almost three decades ago, the algorithm belongs still 
to one of the most important heuristics in this area. However, several basic questions
concerning the algorithm's performance remain open. Using a spectral based approach, 
Boppana constructs an implementable algorithm which, assuming the density difference 
\begin{equation} \label{eq:dens:diff}
  p-q \ge c\sqrt{p\ln n}/\sqrt{n} \quad \text{for a certain  constant $c>0$}
\end{equation}
bisects $\cG_{n}(p,q)$ optimally w.h.p. (certifying the optimality
of the solutions).
Remarkably, for a long time this was the largest subclass of graphs $\cG_{n}(p,q)$
for which a minimum bisection could be found.
Since under the assumption \eqref{eq:dens:diff} the planted bisection is minimum w.h.p.,
Boppana's algorithm solves the recovery problem for the stochastic block model with two communities. 
Boppana's  algorithm works well also on the regular graph model $\cR_{n}(r,b)$,
assuming that 
\begin{equation} \label{eq:rand:reg:assumption}
  r \ge 6 \quad\text{and} \quad 
  b \le o(n^{1-1/\lfloor (r/2+1)/2 \rfloor} ).
\end{equation}

In this paper we investigate the problem if, under assumption \eqref{eq:dens:diff}, 
Boppana's algorithm works well for the semirandom model. This question was posed 
by Feige and Kilian in \cite{Feige2001} and remained open so far.  
In our work we answer the question affirmatively.
We show also that Boppana's algorithm provides the same results as 
the algorithm proposed currently by Hajek, Wu, and Xu \cite{hajek2016achieving}.
As a consequence we get that Boppana's algorithm achieves the optimal 
threshold for exact recovery in the stochastic block model with 
parameters $p=\alpha\log(n)/n$ and $q=\beta\log(n)/n$.
On the other hand we show some limitations of the algorithm.
One of the main results in this direction is that the density difference 
\eqref{eq:dens:diff} is tight: we prove that if 
$p-q \le o(\sqrt{p\cdot  \ln n}/\sqrt{n})$
then the algorithm fails on $\cG_{n}(p,q)$ w.h.p.

\paragraph*{Our Results.}
The motivation of our research was to systematically explore graph properties 
which guarantee that Boppana's algorithm outputs a  certified optimum bisection.
Due to \cite{Boppana1987} we know that random graphs from 
$\cG_{n}(p,q)$ and $\cR_{n}(r,b)$ satisfy such properties w.h.p.~under 
assumptions \eqref{eq:dens:diff}   and \eqref{eq:rand:reg:assumption} 
on $p,q,r,$ and $b$ as discussed above.
But, as we will see later, the algorithm works well also 
for instances which deviate significantly from such random graphs.

Our first technical contribution is a modification of the algorithm 
to cope with graphs of more than one optimum bisection, like e.g. hypercubes. 
The algorithm proposed originally  
by Boppana does not manage to handle such cases.
Our modification is useful to work on wider classes of graphs.

In this paper we introduce a natural generalization of 
the semirandom model of Feige and Kilian \cite{Feige2001}.
Instead of $\cG_{n}(p,q)$, we start with an arbitrary initial 
graph model $\cG_{n}$, and then apply a sequence of the transformations
by a monotone adversary as in \cite{Feige2001}. We denote such a model 
by $\cA(\cG_{n})$. One of our main positive results is that 
if Boppana's algorithm outputs the minimum-size bisection for 
graphs in  $\cG_n$ w.h.p., then the algorithm finds a minimum bisection w.h.p. 
for the adversarial graph model $\cA(\cG_n)$, too.
As a corollary, we get that under assumption \eqref{eq:dens:diff}, 
Boppana's algorithm works well in the semirandom model, denoted here 
as $\cA(\cG_{n}(p,q)$), and, 
assuming \eqref{eq:rand:reg:assumption}, in 
$\cA(\cR_n(r,b))$ \textendash\ the semirandom regular model.
This solves the open problem posed by Feige and Kilian in  \cite{Feige2001}.
To the best of our knowledge,  Boppana's algorithm is the only method known so far,
that finds (w.h.p.) provably optimum bisections  on all of the above random graph classes.

Since the behavior of the
algorithm on the (common) semirandom 
model $\cA(\cG_{n}(p,q))$ remained unknown so far, 
Feige and Kilian proposed in \cite{Feige2001} a new 
semidefinite programming (SDP) based approach which works 
for semirandom graphs, assuming \eqref{eq:dens:diff}. 
The relationship between the performance of the SDP based algorithm 
and Boppana's approach was left in \cite{Feige2001} as an open problem.
Feige and Kilian conjecture 
that for every graph $G$, their objective function $h_p(G)$
to certify the bisection optimality
and the lower bound computed in Boppana's algorithm
give the same value.  
In our paper we answer this question affirmatively.
To compare the algorithms, we provide a primal SDP formulation 
for Boppana's approach and prove that it is equivalent to the dual 
SDP of Feige and Kilian. 
Next we give a dual program to the primal formulation of Boppana's algorithm
and prove that the optima of the primal and dual programs are equal to 
each other. Note that unlike linear programming,
for semidefinite programs there may be a duality gap.
Thus, we show that the bisection algorithm of Feige and Kilian
provides exactly the same results as Boppana's algorithm. 
However, an important advantage of the spectral method by Boppana over
the SDP based approach by Feige and Kilian is that the spectral 
method is practically implementable reducing the bisection problem
for graphs with $n$ vertices to computing minima of a convex function 
of $n$ variables while the
algorithm in \cite{Feige2001} needs to 
solve a semidefinite program over $n^2$ variables.

From the result that the method by Feige and Kilian is equivalent to Boppana's
we get, as a consequence, that Boppana's algorithm achieves the sharp threshold
for exact cluster recovery in the stochastic block model which has been obtained recently 
by Abbe et~al.~\cite{abbe2016exact} and independently by Mossel et~al.~\cite{mossel2015consistency}.
In~\cite{abbe2016exact,mossel2015consistency} it is proved that 
in the (binary) stochastic block model, with $p=\alpha\log(n)/n$ and
$q=\beta\log(n)/n$ for fixed constants $\alpha \not= \beta$, 
if $(\sqrt{\alpha}-\sqrt{\beta})^2>2$, the planted clusters can be
exactly recovered (up to a permutation of cluster indices) with probability converging
to one; if $(\sqrt{\alpha}-\sqrt{\beta})^2<2$, no algorithm can exactly recover 
the clusters with probability converging to one.
Note, that the choice of $p$ and $q$ is well justified: Mossel et~al.
show that if $q<p = \log(n)/n$ then the exact recovery is impossible
for these parameters. 
In \cite{hajek2016achieving} Hajek et~al.
proved that the SDP of Feige and Kilian achieves the optimal threshold, i.e. 
if  $(\sqrt{\alpha}-\sqrt{\beta})^2>2$ then the SDP reconstructs communities w.h.p.
From our result we get, that Boppana's algorithm achieves the threshold, too.

To analyze limitations of the spectral approach
we provide structural properties of the space of feasible solutions 
searched by the algorithm.
This allows us to prove that if an optimal bisection contains some 
forbidden subgraphs, then Boppana's algorithm fails.
Using these tools, we were able to show that 
if the density difference 
$p-q$ is asymptotically smaller than $\sqrt{p\cdot  \ln n}/\sqrt{n}$
then Boppana's  algorithm fails to determine a certified optimum bisection 
on $\cG_{n}(p,q)$ w.h.p. Note that our impossibility result is not a direct consequence 
of the lower bound for the  exact cluster recovery discussed above. 
For example, for $q=\cO(1)/n$ and $p=\sqrt{\log{n}} /n$ 
from Mossel et~al.~\cite{mossel2015consistency} we know that for these parameters 
the exact recovery is impossible but obviously this does not 
imply that determining of a certified optimum bisection is impossible either.

\paragraph*{Related Works.}
Spectral partitioning goes back to Fiedler \cite{Fiedler1975}, who
first proposed to use eigenvectors to derive partitions. Spielman and
Teng e.g. showed, that spectral partitioning works well on
planar graphs \cite{Spielman1996,Spielman2006}, although there are also
graphs on which purely spectral algorithms perform poorly, as shown by
Guattery and Miller \cite{Guattery1998}.

Also other algorithms have been proven to work on the planted
bisection model. Condon and Karp \cite{condon2001algorithms}
developed a linear time algorithm for the more general $l$-partitioning
problem. Their algorithm finds the optimal partition with probability
$1-\exp(-n^{\Theta(\varepsilon)})$ in the planted bisection model with
parameters satisfying $p-q=\Omega(1/n^{1/2-\varepsilon})$. Carson and
Impaglizzo \cite{carson2001hill} show that a hill-climbing algorithm is
able to find the planted bisection w.h.p.
for parameters
$p-q=\Omega((\ln^3 n)/n^{1/4})$.
Dyer and Frieze \cite{dyer1989solution} provide a min-cut via degrees
heuristic that, assuming $n(p-q)=\Omega(n)$ finds and certifies the minimum bisection w.h.p. 
Note, that the density difference~\eqref{eq:dens:diff} assumed by Boppana 
still outperforms the above ones.
Moreover a disadvantage of the methods against Boppana's algorithm, except for the last one,  
is that they do not certify the optimality of the solutions.
In \cite{mcsherry2001spectral} McSherry describes a spectral based heuristic
that applied to $\cG(p,q)$ finds a minimum bisection w.h.p if
$p$ and $q$ satisfy assumption \eqref{eq:dens:diff} but it does not 
certify the optimality.
Importantly, the algorithms above, similarly as Boppana's method,
solve the recovery problem for the stochastic block model with two
communities.

In \cite{Coja2005}  Coja-Oghlan developed a new  spectral-based
algorithm which, on the planted partition model $\cG_{n}(p,q)$,
enables for a wider range of parameters than  \eqref{eq:dens:diff},
certifying the optimality of its solutions.
The algorithm~\cite{Coja2005}
assumes that 
 $ p-q \ge \Omega(\sqrt{p\ln (np)}/\sqrt{n})$. If the parameters $p$ and $q$
describe non-sparse graphs, this condition is essentially the same 
as Boppana's assumption.
For sparse graphs, however, Coja-Oghlan's constraint allows a
larger subclass. 
For example, the algorithm works in $\cG_n(p,q)$  
for $q=\cO(1)/n$ and $p=\sqrt{\log{n}} /n$. Due to 
results presented in our paper we know that 
Boppana's algorithm fails w.h.p. for such graphs. 
Interestingly, the condition on the density difference by Coja-Oghlan
allows graphs for which the minimum bisection width is strictly smaller 
than the width of the planted bisection w.h.p.
However, a drawback of 
Coja-Oghlan's algorithm is that to work well in the planted bisection model with 
\emph{unknown} parameters $p$ and $q$, the algorithm 
has to learn the parameters
since it is based on the knowledge of values $p$ and $q$. Also the performance of the algorithm
on other families, like e.g. semirandom graphs and  the regular random graphs
$\cR_n(r,b)$, is unknown.
Recent research by Coja-Oghlan et\,al. \cite{Coja2015} contributes to a better
understanding of the planted bisection model and average
case behavior of a minimum bisection. 

\vspace*{1mm}
The paper is organized as follows. The next section contains an overview
over Boppana's algorithm. In Section~\ref{sec:non-unique:bisections}
we propose a modification of the algorithm to deal with non-unique optimum bisections.
In Section~\ref{sec:adversarial:model} we define the adversarial graph model and show, that
Boppana's algorithm works well on this class. Next
we develop a new analysis of the algorithm and use it to show some limitations
of the method.
Finally, in Section~\ref{sec:comparing:boppana:feige}
we compare the algorithm to the SDP approach of Feige and Kilian.
We conclude the paper with a discussion.
The  proofs of most of the propositions presented in 
Sections \ref{sec:Boppana:recall} through \ref{sec:comparing:boppana:feige}
are moved to the appendix (Section~\ref{sec:appendix}).

\section{Boppana's  Graph Bisection Algorithm}\label{sec:Boppana:recall}
In this section  we fix definitions and notations 
used in our paper and we recall Boppana's algorithm
and known facts on its performance.  
We need the details of the algorithm to describe 
its extension in the next section. 
For a given  graph $G=(V,E)$, with $V=\{1,\ldots,n\}$, Boppana defines a function $f$
for all  real vectors $x,d\in\mathbb{R}^n$ as
\begin{equation}\textstyle
f(G,d,x)=\sum_{\{i,j\}\in E}\frac{1-x_ix_j}{2}+\sum_{i\in V}d_i (x_i^2-1).
\label{function_f}
\end{equation}
Call by $S\subset \mathbb{R}^n$ the subspace of all vectors 
$x\in\mathbb{R}^n$, with $\sum_i x_i=0$.
Based on $f$, the function $g'$ is defined as follows
\begin{equation}
g'(G,d)=\min_{\|x\|^2=n, x\in S} f(G,d,x),
\end{equation}
where $\|x\|$ denotes $L_2$ norm of $x$.
Vector $x$ is named a \emph{bisection vector} if $x\in \{+1,-1\}^n$ and $\sum_i x_i =0$.
Such $x$ determines a bisection  of $G$ 
of the cut width denoted as $\cutwidth(x)=\sum_{\{i,j\}\in E}\frac{1-x_ix_j}{2}$. 
For a bisection vector $x$ the function $f$ takes the value
\eqref{function_f} regardless of $d$.
Minimization over all such $x$ would give the minimum bisection width. 
Since $g'$ uses a relaxated constraint we get 
$g'(G,d)\leq \bw(G)$ where, recall, $\bw(G)$ 
denotes the
bisection width of $G$.
To improve the bound, Boppana tries to find some $d$ 
which leads to a minimal decrease of the function value of $g'$ 
compared to the bisection width:
\begin{equation}\label{def:function:h}
  h(G)=\max_{d\in \mathbb{R}^n} g'(G,d).
\end{equation}
It is easy to see that for every graph $G$ we have $h(G) \le \bw(G)$.

In order to compute $g'$ efficiently, Boppana expresses the function 
in spectral terms. To describe this we need some definitions.
Let $I$ denote the 
$n$-dimensional identity matrix and let $P=I-\frac{1}{n} J$
be the projection matrix which projects a vector $x\in\mathbb{R}^n$ to 
the projection $Px$ of vector $x$ into the subspace $S$. 
Here, $J$ denotes an $n\times n$ matrix of ones.
For a matrix $B\in \mathbb{R}^{n\times n}$, the matrix $B_S=PBP$ 
projects a vector $x\in\mathbb{R}^n$ to $S$, then applies $B$ 
and projects the result again into $S$.
Further, for $B\in\mathbb{R}^{n\times n}$ and $d\in \mathbb{R}^{n}$ 
we denote the sum of $B$'s elements as $\tsum(B)=\sum_{ij} B_{ij}$ 
and by $\diag(d)$ we denote the $n\times n$ diagonal matrix $D$ 
with the entries of the vector $d$ on the main diagonal, i.\,e. $D_{ii}=d_i$.

Now assume $B\in\mathbb{R}^{n\times n}$ is symmetric and let $B_S=PBP$. 
Denote by $\rwoconst$ the real space $\mathbb{R}^n$ 
without the subspace spanned by the identity vector $\myvec{1}$, i.\,e. 
$\rwoconst=\mathbb{R}^n\setminus \{c\myvec{1}: c\in\mathbb{R}\}$. We define
$\lambda(B_S)=\max_{x\in\rwoconst} \frac{x^T B_S x}{\|x\|}.$
It is easy to see that if $\lambda(B_S)\geq 0$ then
\begin{equation}\label{eq:Rayleigh:Quotient}
\lambda(B_S)=\max_{x\in\mathbb{R}^n} \frac{x^T B_S x}{\|x\|}
\end{equation}
i.\,e. $\lambda(B_S)$ is the largest eigenvalue of the matrix $B_S$.
Vectors $x$ that attain the maximum are exactly the eigenvectors corresponding to the largest eigenvalue $\lambda(B_S)$ of $B_S$.

Let $G$ be an undirected graph with $n$ vertices and adjacency matrix $A$. 
Let further $d\in\mathbb{R}^n$ be some vector and let $B=A+\diag(d)$, then we define
\[g(G,d)=\frac{\tsum(B)-n \lambda(B_S)}{4}.\]
In \cite{Boppana1987} it is shown  that
function $g'$ can be expressed as  $g'(G,d)=g(G,-4d)$.
Since in the definition of $h$ in~\eqref{def:function:h} we maximize over all $d$, we   
can conclude that 
\begin{equation}\label{eq:for:function:h}
  h(G)\ = \ \max_{d\in \mathbb{R}^n} g(G,d)  \ = \ 
      \max_{d\in \mathbb{R}^n}  \frac{\tsum(A+\diag(d))-n \lambda((A+\diag(d))_S)}{4}. 
\end{equation}
Boppana's algorithm that finds and certifies an optimal bisection, works as follows:

\begin{algorithm}[h]
 \SetKwInOut{Input}{Input}
 \caption{Boppana's Algorithm 
        \label{alg:Boppana}}
  \Input{Graph $G$ with adjacency matrix $A$.}
  Compute $h(G)$: Numerically find a vector $d^\opt$ 
  which maximizes $g(G,d)$. Let $D=\diag(d^\opt)$. 
  Use constraint $\sum_i d^\opt_i=2|E|$ to ensure $\lambda((A+D)_S)>0$\;\label{bb:alg:step:one}
  Construct a bisection: 
  Let $x$ be an eigenvector corresponding to the eigenvalue $\lambda((A+D)_S)$. Construct a bisection vector $\hat x$ by splitting at the median $\bar x$ of $x$, i.e. 
  let $\hat x_i = +1$ if $x_i\geq \bar x$ and $\hat x_i = -1$ if $x_i<\bar x$. 
  If $\sum_i {\hat x}_i > 0$, move (arbitrarily)
  $\frac{1}{2}\sum_i {\hat x}_i$ vertices $i$ with $x_i=\bar x$ to part $-1$ letting $\hat x_i = -1$\;\label{bb:alg:step:two}
    Output $\hat x$; If $\cutwidth(\hat x) = h(G)$ output ``optimum bisection'' else output ``fail''.
\end{algorithm}

One can prove that $g$ is concave and hence, the maximum in Step 1 can be 
found in polynomial time with arbitrary precision \cite{Groetschel1981}.
To analyse the algorithm's performance, Boppana proves the following, for a sufficiently large constant
$c>0$:

\begin{theorem}[Boppana \cite{Boppana1987}] \label{thm:boppana:main}
Let $G$ be a random graph from $\cG_{n}(p, q)$, 
and let
$p-q  \ge c (\sqrt{p\ln n}/\sqrt{n})$.
Then with probability 
$1 - \cO(1/n)$, the bisection width of $G$ equals $h(G)$.
\end{theorem}

From this result one can conclude that the value $h(G)$ 
computed by the algorithm is,  w.h.p.,
equal to the optimal bisection width of $G$.
However, to guarantee that the algorithm works well
one needs additionally to show that it also finds an optimal bisection:

\begin{theorem}\label{thm:boppana:certify}
For random graphs $G$ from  $\cG_{n}(p, q)$,
with $p-q\ge c (\sqrt{p\ln n}/\sqrt{n})$,
Boppana's algorithm certifies the optimality of $h(G)$ 
revealing  
w.h.p.
the bisection vector $\hat x$ of $\cutwidth(\hat x)=h(G)$.
\end{theorem}

To prove this theorem one first has to revise carefully the proof  
of Theorem~\ref{thm:boppana:main} in \cite{Boppana1987} 
and show that  w.h.p.
the multiplicity of the 
largest eigenvalue of the matrix $(A+D)_S$ in Step 1 is~1.
This was observed already in \cite{Blumofe1993}.
Next we need  the following property:
\begin{lemma}
Let $G$ be a graph with $h(G)=\bw(G)$ and 
let $d^\opt\in\mathbb{R}^n$ s.\,t. 
$g(G,d^\opt)=\bw(G)$ and $\sum_i d^\opt_i\geq 4\bw(G)-2|E|$. 
Denote further by $B^\opt=A+\diag(d^\opt)$. 
Then every optimum bisection vector $y$ is an eigenvector 
of $B^\opt_S$ corresponding to the largest eigenvalue $\lambda(B^\opt_S)$.
\label{lemma:y:largest:eigenvector}
\end{lemma}
(The proof of Lemma~\ref{lemma:y:largest:eigenvector}, as the 
proofs of most of the remaining propositions presented in this paper,
are given in Section~\ref{sec:appendix}.)
This completes the proof that
the algorithm works well 
on random graphs from  $\cG_{n}(p, q)$.

\section{Certifying Non-Unique Optimum Bisections}
\label{sec:non-unique:bisections}
From the previous section we know that if the bound $h(G)$ 
is tight and the bisection of minimum size is unique, or more precisely the multiplicity of
the largest eigenvector of $B_S$ is~1, Boppana's algorithm 
is able to certify the optimality of the resulting bisection.
We say that a graph $G$ has a unique optimum bisection if there 
exists a unique, up to the sign, bisection vector $x$ such that $\cutwidth(x)=\cutwidth(-x)=\bw(G)$.
In this paper we investigate families of graphs, different than random graphs 
$\cG_{n}( p, q)$, for which the Boppana's approach works well. 
To this aim we first need to show a modification which handles cases such that 
$h(G)=\bw(G)$ but for which no unique bisection of minimum size exists.
As we will see later hypercubes satisfy these two conditions.
We present our algorithm below. 
Note that if the multiplicity of  the largest eigenvalue of
$B^\opt_S$ is 1, then the algorithm outputs the same result as in 
the original algorithm by Boppana.

\begin{algorithm}[h]
 \SetKwInOut{Input}{Input}
 \caption{Boppana's Algorithm Certifying Non-Unique Optimum Bisections
        \label{alg:modif:Boppana}}
  \Input{Graph $G$ with adjacency matrix $A$.}
   Perform Step 1 of Algorithm~\ref{alg:Boppana};
   Let $x$ be an eigenvector corresponding to the eigenvalue $\lambda((A+D)_S)$ and 
   let $k$ be the multiplicity of the largest eigenvalue of $(A+D)_S$\; 
   If $k=1$ then construct a bisection vector $\hat x$ by splitting at the median $\bar x$ 
   as in Step 2  of Algorithm~\ref{alg:Boppana};
   Next output $\hat x$ and  
   if $\cutwidth(\hat x) = h(G)$ output ``optimum bisection'' else output ``fail'';
   If  $k>1$ then perform the steps below\;
   Let $M\in\mathbb{R}^{n\times k}$ be the matrix with 
   $k$ linear independent 
   eigenvectors corresponding to this largest eigenvalue;
   Transform the matrix to the reduced column echelon form, i.\,e.\ there are $k$ 
   rows which form an identity matrix, s.t. $M$ still spans the same subspace\;
   Brute force: for every combination of $k$ 
   coefficients from $\{+1,-1\}$ take the linear combination of 
   the $k$ vectors of $M$ with the coefficients and
   verify if the resulting vector $x$ is a bisection vector, i.e. $x\in \{+1,-1\}^n$
   with $\sum_i x_i=0$. If yes and if $\cutwidth(x)=h(G)$ then output $x$ and continue.
   This needs $2^k$ iterations\; \label{mod:boppana:brute:force} 
    If in Step~\ref{mod:boppana:brute:force} no bisection vector $x$
    is given then output ``fail''.
\end{algorithm}

\begin{theorem}
  If $h(G)=\bw(G)$ then  the algorithm above reconstructs all optimal bisections. 
  Every achieved bisection vector corresponds to an optimal bisection.
\label{theorem:bop:modified:works}
\end{theorem}

The eigenvalues for the family  of hypercubes are explicitly known \cite{harary1988}.
Hence, we can
verify that the bound $h(G)$ is tight and Boppana's algorithm 
with the modification above works, i.e. finds an optimal bisection. For a hypercube $H_n$ with
$n$ vertices we have $h(H_n)=g(H_n,$ $(2-\log n)\myvec{1})=n/2=\bw(H_n)$.
Since the hypercube with $n$ vertices has $\log n$ optimal bisections and
the largest eigenspace of $B_S$ has multiplicity $\log n$, the brute force part in our
modification of Boppana's algorithm results in a linear factor of $n$ for the overall runtime.
Thus, the algorithm runs in polynomial time.
In the next section we will extend this result to an adversarial model based
on hypercubes and show, that Boppana's algorithm works on that model as well.

\section{Bisections in Adversarial Models}
\label{sec:adversarial:model}

We introduce the \emph{adversarial model}, denoted by $\cA(\cG_n)$,
as a generalization of the semirandom model in the following way.
Let $\cG_n$ be a graph model, i.e. a class of graphs 
with distributions over graphs of $n$ nodes ($n$ even).
In the model $\cA(\cG_n)$, initially a graph $G$ is chosen at random according to
$\cG_n$. Let $(Y_1,Y_2)$
be a fixed, but arbitrary optimal bisection of $G$.
Then, similarly as in \cite{Feige2001}, a monotone adversary is allowed to modify
$G$ by applying an arbitrary sequence of the following monotone transformations:
\begin{compactenum}
  \item The adversary may remove from the graph any edge $\{u, v\}$
crossing a minimal bisection ($u \in Y_1$ and $v \in Y_2$);
  \item The adversary may add to the graph any edge $\{u, v\}$ not crossing
the bisection ($u,v \in Y_1$ or $u,v \in Y_2$).
\end{compactenum}
For example, $\cA(\cG_n(p,q))$ is the semirandom model as defined in \cite{Feige2001}.

We will prove that Boppana's algorithm works well 
for graphs from adversarial model $\cA(\cG_n)$ if the algorithm
works well for $\cG_n$.
First we show that, if the algorithm is able 
to find an optimal bisection size of a graph, we can add edges 
within the same part of an optimum bisection and that 
we can remove cut edges, and the algorithm will still work.
This solves the open question of Feige and Kilian \cite{Feige2001}.

Note that the result follows alternatively from 
Corollary~\ref{corr:feige:is:boppana}
(presented in Section~\ref{sec:comparing:boppana:feige})
that the SDPs  of \cite{Feige2001} are equivalent to Boppana’s optimization function 
and form the property proved in \cite{Feige2001} that 
the objective function of the dual SDP of Feige and Kilian
preserves minimal bisection regardless of monotone transformations. 
The aim of this section is to give a direct proof of this property for Boppana's algorithm.

\begin{theorem}
Let $G=(V,E)$ be a graph with $h(G)=\bw(G)$. Consider some optimum bisection
$Y_1,Y_2$ of $G$.
\begin{compactenum}
\item Let $u$ and $v$ be two vertices within the same part, i.e. $u,v\in Y_1$ or $u,v\in Y_2$,
  and let $G'=(V, E\cup \{\{u,v\}\})$. Then $h(G')=\bw(G')$.
\item Let $u$ and $v$ be two vertices in different parts, i.e. $u\in Y_1$ and $v\in Y_2$, with $\{\{u,v\}\}\in E$
  and let $G'=(V, E\setminus \{\{u,v\}\})$. Then $h(G')=\bw(G)-1=\bw(G')$.
\end{compactenum}
  \label{theorem:add:edge:within:part}
  \label{theorem:remove:cut:edge}
\end{theorem}

\begin{proof}[Sketch of proof]
In order to prove the first part of the theorem, i.e. when  
we add an edge~$\{u,v\}$, let $A$ and $A'$ denote the adjacency matrices of $G$ and $G'$, 
respectively. It holds $A'=A+A^\Delta$ with $A^\Delta_{uv}=A^\Delta_{vu}=1$ and zero everywhere else. 
The main idea is now, that we can derive a new optimal correction vector $d'$ for $G'$ based on the optimal
correction vector $d^\opt$ for $G$. We set $d'=d^\opt+d^\Delta$ with
\[ 
  d^\Delta_i=\begin{cases}
  -1 & \text{ if } i=u \text{ or } i=v,\\
  0 & \text{ else.}
  \end{cases}
\]

The known changes in the adjacency matrix as well as the derived correction vector allow us to compute $g(G',d')$
and to show that $g(G',d')=\bw(G')$. The proof of the second part of the theorem works analogously.
The complete proof can be found in the appendix.
\end{proof}

\begin{theorem}
  If Boppana's algorithm finds a minimum bisection for a graph model $\cG_n$ w.h.p., then
  it finds a minimum bisection w.h.p. for the adversarial model $\cA(\cG_n)$, too.
\label{theorem:boppana:adversary}
\end{theorem}

As a direct consequence, we obtain the following corollary regarding the semirandom graph
model considered by Feige and Kilian:

\begin{corollary}\label{corr:Bopp:works:for:semi:random}
Under assumption \eqref{eq:dens:diff} on $p$ and $q$, Boppana's algorithm 
computes the minimum bisection in  $\cA(\cG_n(p,q))$, i.e. in the semirandom model,  w.h.p.
\end{corollary}

In \cite{Boppana1987}, Boppana also considers random regular graphs $\cR_{n}(r,b)$, where
a graph is chosen uniformly over the set of all $r$-regular graphs with bisection width~$b$.
He shows that his algorithm works w.h.p. on this graph under the assumption that
$b=o(n^{1-1/\lfloor (r+1)/2\rfloor})$. We can now define the semirandom regular graph model as adversarial model $\cA(\cR_{n}(r,b))$.
 Applying Theorem~\ref{theorem:boppana:adversary}, we obtain

\begin{corollary}
Under assumption \eqref{eq:dens:diff}  on $p$ and $q$, Boppana's algorithm 
computes the minimum bisection in the semirandom regular model  w.h.p.
\end{corollary}

Theorem~\ref{theorem:boppana:adversary} can also be applied on deterministic graph classes,
e.g. the class of hypercubes. We then obtain:
\begin{corollary}
  Boppana's algorithm (with our modification for non-unique bisections) finds an optimal bisection
  on adversarial modified hypercubes.
\end{corollary}

\section{The Limitations of the Algorithm}
\label{sec:limitations}

Boppana shows, that his algorithm works well on some classes of random graphs. 
However, we do not know which graph properties force the algorithm to fail. 
For example, for the considered planted bisection model, 
we require a small bisection width. On the
other hand, as we have seen in Section \ref{sec:non-unique:bisections}  
Boppana's algorithm works for the hypercubes 
and their semirandom modifications \textendash\ graphs that 
have large minimum bisection sizes.

In the following, we present newly discovered structural properties from inside the algorithm,
which provide a framework for a better analysis of the algorithm itself. 
Let $y$ be a bisection vector of $G$. We define
\begin{equation} \label{def:dy}
  d^{(y)}=-\diag(y)Ay.
\end{equation}
An equivalent but more intuitive characterization of $d^{(y)}$ is the following:
$d^{(y)}_i$ is the difference between the number of adjacent vertices in other partition as vertex $i$
and the number of adjacent vertices in same partition as $i$.

\begin{lemma}
  Let $G$ be a graph with $h(G)=\bw(G)$ and assume there is more than one optimum bisection in $G$.
  Then  (up to  constant translation vectors  $c\myvec{1}$)  there exists a 
  unique vector $d^\opt$ with $g(G,d^\opt)=\bw(G)$. 
  Additionally, for every  bisection vector $y$ of an arbitrary optimum bisection in $G$
  there exists a unique
   $\alpha^{(y)}$
  and the corresponding $d^{(y)}$, with  $g(G,d^{(y)}+\alpha^{(y)} y)=\bw(G)$.
  \label{lemma:dopt:unique}
\end{lemma}

Thus, if there are two optimum bisections 
representing by $y$ and $y'$ with $d^{(y)}\ne d^{(y')}$, 
then the difference of the $d$-vectors in component $i$ 
is only dependent on $y_i$ and $y'_i$, since we have $d^{(y)}-d^{(y')}=\beta' y'-\beta y$
for some constants $\beta$ and $\beta'$.
This structural property allows us to show the following limitation for the sparse planted partition model $\cG_{n}(p,q)$.

\begin{theorem}\label{thm:main:neg:res}
  The algorithm of Boppana fails w.h.p. in the subcritical phase from \cite{Coja2005}, defined as
  $n(p-q) = \sqrt{np\cdot \gamma\ln n}$, for real $\gamma>0$.
\end{theorem}

In the planted partition model $\cG_{n}(p,q)$, if the graphs are dense, e.g. 
$p=1/n^c$ for a constant $c$ with $0<c<1$, the constraints for the density difference $p-q$   
assumed in Boppana's  \cite{Boppana1987} and Coja-Oghlan's  \cite{Coja2005} 
algorithms
 are essentially the same. 
However for sparse graphs,  e.g. such that $q=\cO(1)/n$, the situation changes drastically.
Now, e.g. $p=\sqrt{\log{n}} /n$ satisfy Coja-Oghlan's constraint 
$p-q \ge \Omega(\sqrt{p\ln (pn)}/\sqrt{n})$
but the condition on the difference $p-q$ assumed by Boppana is not true any more.
Theorem~\ref{thm:main:neg:res} shows that Boppana's algorithm indeed fails under this setting.

The proof of this theorem relies on the following observation, 
which can be derived from our newly discovered structural properties from above. 

\begin{lemma}
  Let $G$ be a graph with $h(G)=\bw(G)$ and let $(Y_1, Y_{-1})$ be an arbitrary optimal bisection. Then, for each pair of vertices $v_i\in Y_i$, $i\in\{1,-1\}$, not connected by an edge ($\{v_i,v_{-i}\}\not\in E$), we have: If $e(v_i,Y_i)=e(v_i,Y_{-i})$ for $i\in\{1,-1\}$ (the vertices have balanced degree), then $N(v_i)=N(v_{-i})$, i.e. both vertices have the same neighbors.
  \label{lemma:boppana:same:neighbors}
\end{lemma}
I.e. if we have two balanced vertices in different parts of an optimal bisection, not connected by an edge, then the two vertices must have the same neighborhood as a necessary criterion for Boppana's algorithm to work. In the subcritical phase in Theorem~\ref{thm:main:neg:res}, there exist most likely many of such pairs of vertices, but they are unlikely to have all even the same degree.

We can also provide forbidden substructures, which make Boppana's algorithm fail. This is e.g. the case, when the graph contains a path segment located on an optimal bisection:

\begin{corollary}\label{corollary:path}
Let $G$ be a graph, as illustrated in Fig.~\ref{fig:path} (left), with $n\geq 10$ vertices containing a path segment $\{u',u\},\{u,w\},\{w,w'\}$, where $u$ and $w$ have no further edges. If there is an optimal bisection $y$, s.\,t. $y_u=y_{u'}=+1$ and $y_w=y_{w'}=-1$ (i.\,e. $\{u,w\}$ is a cut edge), then $h(G)<\bw(G)$.
\end{corollary}

To prove this corollary, we use the more general but more technical Lemma~\ref{lemma:necessary:many:edges} (in Appendix) with parameters $\tilde C_{+1}=\{u\}$ and $\tilde C_{-1}=\{w\}$.

\begin{figure}[h!]
\begin{center}
  \begin{tikzpicture}[scale=0.8, every node/.style={scale=0.8}]
    \node[draw] (a) at (0,2) {$\ldots u'$};
    \node (b) at (2,2) {$u$};
    \node (c) at (4,2) {$w$};
    \node[draw] (d) at (6,2) {$w' \ldots$};
    \draw (a) -- (b) -- (c) -- (d);
  \end{tikzpicture}
  \hspace{2cm}
  \begin{tikzpicture}[scale=0.6, every node/.style={scale=0.8}]
    \node[draw] (a1) at (0,0) {$\ldots u_1'$};
    \node (b1) at (2,0) {$u_1$};
    \node (c1) at (4,0) {$w_1$};
    \node[draw] (d1) at (6,0) {$w_1' \ldots$};

    \node[draw] (a2) at (0,-2) {$\ldots u_2'$};
    \node (b2) at (2,-2) {$u_2$};
    \node (c2) at (4,-2) {$w_2$};
    \node[draw] (d2) at (6,-2) {$w_2' \ldots$};

    \draw (a1) -- (b1) -- (c1) -- (d1);
    \draw (a2) -- (b2) -- (c2) -- (d2);
    \draw (b1) -- (b2);
    \draw (c1) -- (c2);
    \draw[dashed] (a1) -- (a2);
    \draw[dashed] (d1) -- (d2);
  \end{tikzpicture}
\caption{Forbidden graph structures as in Corollary~\ref{corollary:path} (left) and in Corollary~\ref{corollary:grid} (right).}
\label{fig:path}
\label{fig:grid}
\end{center}
\end{figure}
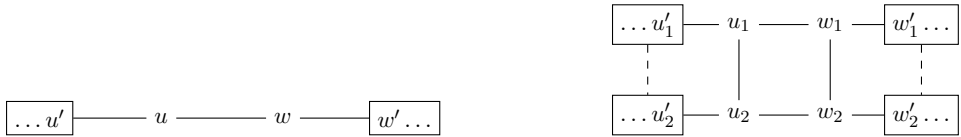

The result can also be applied for $2\times c$ lattices:
\begin{corollary}\label{corollary:grid}
Let $G$ be a graph with $n\geq 10c$ vertices containing a $2\times c$ lattice with vertices $u_i$ and $w_i$, as illustrated in Fig.~\ref{fig:grid} (right). (The construction is similar to the corollary above, but now we have a lattice instead of a single cut edge.) If there is an optimal bisection $y$, s.\,t. $y_{u_i}=y_{u_i'}=+1$ and $y_{w_i}=y_{w_i'}=-1$, then $h(G)<\bw(G)$.
\end{corollary}

Futhermore, the algorithm fails if there are isolated vertices in both parts of an optimal bisection:
\begin{theorem}
  Let $G$ be a graph with $h(G)=\bw(G)$. Let $G'$ be the graph $G$ with two additional isolated vertices, then $h(G')\leq h(G)-\frac{4\bw(G)}{n^2}$. (Note: $G$ has $n$ vertices and $G'$ has $n+2$ vertices.)
  \label{theorem:no:two:isolated}
\end{theorem}

\section{SDP Characterizations of the Graph Bisection Problem}
\label{sec:comparing:boppana:feige}

Feige and Kilian express the minimum-size bisection problem for an instance graph $G$
as a semidefinite programming problem (SDP) 
with solution $h_p(G)$ and prove that the function $h_d(G)$,
which is the solution to the dual SDP,
reaches $\bw(G)$ 
 w.h.p.
Since $\bw(G)\geq h_p(G)\geq h_d(G)$, they conclude that $h_p(G)$ as well
reaches $\bw(G)$ w.h.p.
The proposed algorithm computes $h_p(G)$  
and reconstructs the minimum bisection of $G$ from the optimum solution
of the primal SDP.
The authors conjecture in \cite[Sec.~4.1.]{Feige2001} the following: 
"Possibly, for every graph $G$, the function $h_p(G)$ and 
the lower bound $h(G)$ computed in Boppana's algorithm
give the same value, making the lemma that $h_p(G)=\bw(G)$ w.h.p. a restatement 
of the main theorem of~\cite{Boppana1987}.
In this section we answer this question affirmatively.

The semidefinite programming approach for optimization problems
was studied by Alizadeh~\cite{alizadeh1995interior}, who as first
provided an equivalent SDP formulation of Boppana's algorithm.
Before we give an SDP introduced by Feige an Kilian, 
we recall briefly some basic definitions and provide an SDP 
formulation for Boppana's  approach.
On the space $\mathbb{R}^{n\times m}$ of $n \times m$ matrices,
we denote by $A\bullet B$ an inner product  of $A$ and $B$ defined as
$A\bullet B = \tr(AB)=\sum_{i=1}^n \sum_{j=1}^m A_{ij}B_{ij}$,
where $\tr(C)$ is the trace of the (square) matrix $C$.
Let A be an $n\times n$ symmetric real matrix, then $A$ is called 
symmetric positive semidefinite (SPSD)
if $A$ is symmetric, i.e. $A^T=A$,
and for all real vectors $v\in \mathbb{R}^n$ we have $v^T A  v \ge 0$. 
This property is denoted by $A\succeq 0$. Note that the eigenvalues 
of a symmetric matrix are real.

For given real vector $c\in \mathbb{R}^n$ and $m+1$
symmetric matrices $F_0,\ldots,F_m\in \mathbb{R}^{n\times n}$  
an SDP over variables $x \in \mathbb{R}^n$ is defined as
\begin{equation} \label{eq:def:sdp:primal}
  \min_x c^T x \quad \text{subject to}\quad  
 F_0 +\sum_{i=1}^m x_i F_i\ \succeq \ 0.
\end{equation}
The dual program associated with the SDP 
(for details see  e.g. \cite{vandenberghe1996})
is the program over the variable matrix $Y=Y^T \in \mathbb{R}^{n\times n}$:
\begin{equation} \label{eq:def:sdp:dual}
  \max_Y\  - F_0 \bullet  Y  \quad \text{subject to}\quad  
  \forall i: \ F_i \bullet Y =c_i \quad \text{and}\quad  Y\ \succeq \ 0.
\end{equation}
It is known that the optimal value of the maximization dual SDP
is never larger than the optimal value of the minimization primal counterpart.
However, unlike linear programming, for semidefinite programs
there may be a duality gap, i.e. the primal and/or dual might not 
attain their respective optima.

To prove that for any graph $G$ Boppana's function $h(G)$
gives the same value as $h_p(G)$
we formulate the function $h$ as a (primal) SDP.
We provide also its dual program and prove that the optimum 
solutions of primal and dual are equal in this case.
Then we show that the dual formulation of the Boppana's optimization 
is equivalent to the primal SDP defined by Feige and Kilian~\cite{Feige2001}.

Below,  $G=(V,E)$ denotes a graph, $A$ the adjacency matrix of $G$
and for a given vector $d$, as usually,  let $D=\diag(d)$,  for short.
We provide the SDP for the function $h$ (Eq.~\eqref{eq:for:function:h}) that
differ slightly from that one given in~\cite{alizadeh1995interior}.  

\begin{proposition}\label{prop:Bopp:reformulation}
For any graph $G=(V,E)$, the objective function
\[
  h(G)\ = \ \max_{d\in \mathbb{R}^n} \frac{\tsum(A+D)-n \lambda((A+D)_S)}{4}
\]
maximized by Boppana's algorithm 
can be characterized
as an SDP as follows:
\begin{equation}\label{primal:sdp:boppana}
\left\{ \begin{array}{rcl}
 \multicolumn{3}{l}{ p(G) = {\displaystyle \min_{z\in \mathbb{R},d\in \mathbb{R}^n}}
   (nz-\myvec{1}^T d )  \quad \text{subject to} }\\[3mm]
  zI-A+\frac{JA+AJ}{n}-\frac{\tsum(A)J}{n^2}-D+\frac{\myvec{1} d^T+d \myvec{1}^T}{n}-\frac{\tsum(D)J}{n^2} &\succeq & 0,
\end{array}
\right.
\end{equation}
with the relationship  $h(G)=\frac{|E|}{2}-\frac{1}{4}p(G)$.
The dual program to the program~\eqref{primal:sdp:boppana} can be expressed as follows: 
\begin{equation}\label{dual:sdp:boppana}
\left\{
\begin{array}{rcl}
  \multicolumn{3}{l}{d(G)= {\displaystyle\max_{Y\in \mathbb{R}^{n\times n}}} 
    \left(A\bullet Y-\frac{1}{n}\sum_j \deg(j) \sum_i y_{ij}-\frac{1}{n}\sum_i \deg(i) \sum_j y_{ij}+\frac{1}{n^2}\sum_{i,j}y_{ij}\right)}\\[3mm]
   \multicolumn{3}{l}{\text{subject to}}\\[2mm]
  \sum_i y_{ii}&=&n,\\[2mm]
\forall i \quad  y_{ii}-\frac{1}{n}\sum_j y_{ji}-\frac{1}{n} \sum_j y_{ij}+\frac{1}{n^2} \sum_{k,j} y_{kj}&=&1,\\[2mm]
  Y&\succeq& 0.
\end{array}
\right.
\end{equation}
\end{proposition}

Using these formulations we prove  that the primal and dual SDPs  
attain the same  optima.

\begin{theorem}\label{th:sdp:boppana:dual:primal:eq}
  For the semidefinite programs of Proposition~\ref{prop:Bopp:reformulation} 
  the optimal value $p^*$ of the primal SDP~\eqref{primal:sdp:boppana}
  is equal to the optimal value $d^*$ of the dual SDP~\eqref{dual:sdp:boppana}.
  Moreover, there exists a feasible solution $(z,d)$ achieving the optimal value $p^*$.
\end{theorem}

\begin{proof}
Consider the primal SDP~\eqref{primal:sdp:boppana} of Boppana in the form
\[
  \min_{z\in\mathbb{R}, d\in\mathbb{R}^n} 
  z \quad\text{ s.t. }\quad  zI-M(d) \ \succeq \ 0,  
\]
with $M(d)=P(A+\diag(d))P-\frac{\myvec{1}^Td}{n}I$ and, recall, $P=I-\frac{J}{n}$.
Note that this formulation is equivalent to~\eqref{primal:sdp:boppana}, as we have
shown in the proof of Proposition~\ref{prop:Bopp:reformulation}.
We show that this primal SDP problem is strictly feasible, i.e. that 
there exists an $z'$ and an $d'$ with $z'I-M(d')\succ 0.$
To this aim we choose an arbitrary $d'$ and then some 
$z'>\lambda(M(d'))$.
From \cite[Thm. 3.1]{vandenberghe1996}, 
it follows that the optima of primal and dual obtain the same value.

To prove the second part of the theorem, i.e. there exists a feasible solution achieving
the optimal value $p^*$, consider the following. The function $h(G)$ maximizes $g(G,d)$ 
over vectors $d\in\mathbb{R}^n$, while $d$ can be restricted to vectors of mean zero.
The function $g$ is convex and goes to $-\infty$ for vectors $d$ with some component
going to $\infty$. Thus, $g$ reaches its maximum at some finite $d^\opt$.
Now we choose $d=d^\opt$ and $z=\lambda(M(d^\opt))$. Clearly, this solution is feasible and
obtains the optimal value $p^*$.
\end{proof}

For a  graph $G=(V,E)$, Feige and Kilian express the minimum bisection 
problem as an SDP over an $n\times n$ matrix $Y$ as follows: 
\begin{equation}\label{primal:sdp:feige:kilian}
  h_p(G) = \min_{Y\in \mathbb{R}^{n\times n}}  h_Y(G) \quad\text{ s.t. }\quad
  \ \forall i \ y_{ii}=1, \ \sum_{i,j} y_{ij} =0,  \text{ and } Y\succeq 0,
\end{equation}
where $h_Y(G)=\sum_{\substack{\{i,j\} \in E \\ i<j}} \frac{1-y_{ij}}{2}.$
For proving that the SDP takes as optimum the bisection 
width w.h.p. on $\mathcal{G}_n(p,q)$, the authors consider the dual of their SDP:
\begin{equation}\label{dual:sdp:feige:kilian}
  h_d(G)=\max_{x\in \mathbb{R}^{n} } \left(\frac{|E|}{2}+\frac{1}{4}\sum_i x_i\right)
  \quad\text{ s.t. }\quad
  M=-A-x_0 J-\diag(x) \succeq 0,
\end{equation}
where $A$ is the adjacency matrix of $G$.
They show that the dual takes the value of the bisection width w.h.p. and bounds the optimum of the primal SDP. Although we know that their SDP and Boppana's algorithm both work well 
on $\mathcal{G}_n(p,q)$, 
it was open so far how they are related to each other.
Below we answer this question showing that the formulations are equivalent.
We start with the following:

\begin{theorem}
  The primal SDP~\eqref{primal:sdp:feige:kilian}
  is equivalent 
  to the dual SDP~\eqref{dual:sdp:boppana}, with the relationship 
  $h_p(G)=\frac{|E|}{2} -\frac{1}{4}d(G)$.
\label{theorem:feige:is:dual:boppana}
\end{theorem}

From  Theorems  \ref{th:sdp:boppana:dual:primal:eq} and \ref{theorem:feige:is:dual:boppana} we get 

\begin{corollary}\label{corr:feige:is:boppana}
  Let $G$ be an arbitrary graph. Then for the lower bound $h(G)$ 
  of Boppana's algorithm and for the objective functions $h_p(G)$
  of the primal SDP~\eqref{primal:sdp:feige:kilian}, resp. $h_d(G)$ 
  of the dual SDP~\eqref{primal:sdp:feige:kilian}
  of Feige and Kilian 
  \cite{Feige2001} it is true
  \[ h(G)= h_p(G)=h_d(G).\]
\end{corollary}

Thus, the both algorithms provide for any graph $G$ the same objective value.
We want to point out another important fact: 
the bisection algorithm proposed in~\cite{Feige2001} use an SDP
 formulation, where the variables are a matrix with dimension $n\times n$. Thus, there are
 $n^2$ variables for a graph with $n$ vertices.
In contrast, Boppana's algorithm uses $n$ variables in the convex optimization problem.
If we consider the dual SDP, we again have only $n+1$ variables. 
However, due to Corollary~\ref{corr:feige:is:boppana},
we can't be better than Boppana's algorithm.

Abbe et~al.~\cite{abbe2016exact} and independently 
Mossel et~al. \cite{mossel2015consistency}
have shown, that there is a sharp threshold
phenomenon when considering the $\cG_n(p,q)$ model with $p=\alpha\log(n)/n$ and
$q=\beta\log(n)/n$ for fixed constants $\alpha,\beta$, $\alpha>\beta$. Exact recovery
of the planted bisection is possible if and only if $(\sqrt{\alpha}-\sqrt{\beta})^2>2$
(see e.g. \cite{mossel2015consistency} for a formal definition of exact cluster recovery problem).
Hajek et~al.~\cite{hajek2016achieving} show, than an SDP equivalent to
the one of Feige and Kilian achieves this bound. Since, due to 
Corollary~\ref{corr:feige:is:boppana}, we know that
the SDP is equivalent to Boppana's algorithm, we conclude that also Boppana's algorithm
achieves the optimal threshold for finding and certifying the optimal bisection in the
considered model. We get:
\begin{theorem}
  Let $\alpha$ and $\beta$, $\alpha>\beta$, be constants. Consider the graph model
$\cG_n(p,q)$ with $p=\alpha\log(n)/n$ and $q=\beta\log(n)/n$. Then, as $n\to\infty$,
if $(\sqrt{\alpha}-\sqrt{\beta})^2>2$, Boppana's algorithm recovers the planted bisection
w.h.p. If $(\sqrt{\alpha}-\sqrt{\beta})^2<2$, no algorithm is able to recover the planted
bisection w.h.p.
\end{theorem}
\begin{proof}
  The second part of the theorem is exactly the statement from \cite{abbe2016exact}. The first part, i.e. that Boppana's algorithm is able to recover the bisection, follows from \cite[Thm. 2]{hajek2016achieving}. 
  Hajek et al.  show, that for $(\sqrt{\alpha}-\sqrt{\beta})^2>2$ the SDP of Feige and Kilian obtain the optimal solution. Due to Theorem \ref{theorem:feige:is:dual:boppana}, the same holds for Boppana's algorithm.
\end{proof}

\section{Discussion and Open Problems}

Boppana's spectral method is a practically implementable heuristic. Computing eigenvalues and eigenvectors is well-studied and can be done very efficiently. Falkner, Rendl and Wolkowicz  \cite{falkner1994computational} show in a numerical study that using spectral techniques for graph partitioning is very robust and upper and lower bounds for the bisection width can be obtained such that the relative gap is often just a few percentage points apart. In \cite{tu2000algorithms} and \cite{tu1998spectral}, Tu, Shieh and Cheng present numerical experiments including results for Boppana's algorithm. They verify that the algorithm indeed has good average case behavior over certain probability distributions on graphs.
We conducted further experiments on the graph model $\cR_n(r,b)$ which
indicated, that Boppana's algorithm also works for $r=5$, but not for
$r=3$ and $r=4$. An interesting question arising is, which properties of 3- and
4-regular graphs from the planted bisection model let the algorithm fail.

\bibliography{literature}

\section{Proofs}\label{sec:appendix}

\subsection*{Proof of Lemma~\ref{lemma:y:largest:eigenvector}}
We know
\begin{align*}
 && g(G,d^\opt)=\bw(G) &=\frac{\tsum(B^\opt)-n \lambda(B^\opt_S)}{4} \eqcomment{B^\opt=A+\diag(d^\opt)}\\
\Leftrightarrow && \lambda(B^\opt_S) &=\frac{\tsum(B^\opt)-4\bw(G)}{n} \eqlabel{tmp:y:ev:lambda}\\
&&&\geq\frac{\tsum(A)+4\bw(G)-2|E|-4\bw(G)}{n}=0 \eqcomment{\sum_i d^\opt_i\geq 4\bw(G)-2|E|}
\end{align*}
Thus, we conclude that $\lambda(B^\opt_S)\geq 0$.

We compute the value of the Rayleigh quotient of $B^\opt_S$ 
and the optimum bisection vector $y$:
\begin{align*}
  \frac{y^T B^\opt_S y}{\|y\|^2} &= \frac{y^T PB^\opt P y}{\|y\|^2}= \frac{y^T B^\opt y}{n} \eqcomment{y \text{ has mean zero}}\\
&= \frac{y^T (A+\diag(d^\opt)) y}{n}= \frac{y^T A y+\sum_i d^\opt_i}{n} \eqcomment{y_i^2=1 \text{ on the diagonal}}
\end{align*}
We have $y^TAy=\sum_{i,j} A_{ij}y_i y_j$.
According to the definition $A_{ij}=1$ if there is an edge $\{i,j\}\in E$. Edges with both vertices in the same part contribute (twice) by $1$ to the sum. Cut edges on the other hand contribute (twice) by $-1$. There are $\bw(G)$ cut edges. Hence, $y^TAy=\tsum(A)-4\bw(G)$
and we get:
\[\frac{y^T A y+\sum_i d^\opt_i}{n}=\frac{\tsum(A)-4\bw(G)+\sum_i d^\opt_i}{n}=\frac{\tsum(B^\opt)-4\bw(G)}{n}\overset{\eqref{tmp:y:ev:lambda}}{=}\lambda(B^\opt_S).\]

Since the Rayleigh quotient of $B^\opt_S$ and $y$  
takes the value $\lambda(B^\opt_S)$ and $\lambda(B^\opt_S)\geq 0$, we conclude 
that $y$ is an eigenvector of $B^\opt_S$ corresponding to the eigenvalue $\lambda(B^\opt_S)$.\qed

\subsection*{Proof of Theorem~\ref{theorem:bop:modified:works}}
Due to Lemma~\ref{lemma:y:largest:eigenvector}, 
all optimum bisection vectors $y$ are found in this subspace. 
We show even more, namely that 
non-optimum bisection vectors $y$ are not in this subspace. 
For contradiction, assume $y$ is eigenvector. Consider
\begin{align*}
  \lambda(B^\opt_S)&=\frac{y^TB^\opt_Sy}{\|y\|}=\frac{y^TB^\opt y}{\|y\|}=
  \frac{y^T(A+\diag(d^\opt))y}{\|y\|}
  =\frac{y^TAy+\sum_i d^\opt_i}{\|y\|}.
\end{align*}
For a bisection vector $y$, the value $y^TAy$ 
counts the number of cut edges. 
Since this has been minimized, $y$ has to be an optimal bisection.\qed

\subsection*{Proof of Theorem~\ref{theorem:add:edge:within:part}}
We start by proving the first part, i.e. when  
we add an edge~$\{u,v\}$. Let $A$ and $A'$ denote the adjacency matrices of $G$ and $G'$, 
respectively. It holds $A'=A+A^\Delta$ with 
$A^\Delta_{uv}=A^\Delta_{vu}=1$ and zero everywhere else. 
Since $h(G)=\bw(G)$, there exists a $d^\opt$ with 
$g(G,d^\opt)=\bw(G)$. For $G'$, we set $d'=d^\opt+d^\Delta$ with
\[ 
  d^\Delta_i=\begin{cases}
  -1 & \text{ if } i=u \text{ or } i=v,\\
  0 & \text{ else.}
  \end{cases}
\]

W.l.o.g. we restrict ourselves to solutions, with  
$\sum_i d^\opt_i=4\bw(G)-2|E|$ and hence have $\lambda(B^\opt_S)=0$ 
where $B^\opt=A+\diag(d^\opt)$. 
Since $\sum_i d'_i=4\bw(G)-2|E|-2=4\bw(G)-2|E'|$, we want to show 
that $\lambda(B'_S)=0$ holds, where $B'=A'+\diag(d')$. Since 
$B'
=A+A^\Delta+\diag(d^\opt+d^\Delta)=B^\opt+A^\Delta+\diag(d^\Delta)$, we get
\begin{align*}
\lambda(B'_S)&=\max_{x\in\swozero} \frac{x^T (B^\opt+A^\Delta+\diag(d^\Delta)) x}{\|x\|^2}\\
  &=\max_{x\in\swozero} \frac{x^T B^\opt x +x^T(A^\Delta+\diag(d^\Delta)) x}{\|x\|^2}\\
  &=\max_{x\in\swozero} \frac{x^T B^\opt x +2x_ux_v-x_u^2-x_v^2}{\|x\|^2}\\
  &=\max_{x\in\swozero} \frac{x^T B^\opt x -(x_u-x_v)^2}{\|x\|^2}\\
  &\leq \max_{x\in\swozero} \frac{x^T B^\opt x}{\|x\|^2}=0.
\end{align*}
For the bisection vector of an minimal cut size,  
we have $x_u=x_v=1$ or $x_u=x_v=-1$ and thus the 
last inequality is equality. 
Hence, $\lambda(B'_S)=0$ and $g(G',d')=\bw(G')$. This 
completes the proof for the first part.

The proof for the second part is similar to the above one. 
Assume $\{u,v\}$, with $u\in Y_1$ and $v\in Y_2$, is a removed edge from $G$.
We define $d^\Delta$ as we have done above and we let $A'=A+A^\Delta$, 
with $A^\Delta_{uv}=A^\Delta_{vu}=-1$ and zero everywhere else.
Since $\sum_i d'_i=4\bw(G)-2|E|-2=4(\bw(G)-1)-2(|E|-1)=4\bw(G')-2|E'|$, 
our aim is  to show that $\lambda(B'_S)=0$ holds with $B'=A'+\diag(d')$. Indeed we have:
\begin{align*}
\lambda(B'_S)&=\max_{x\in\swozero} \frac{x^T B^\opt x +x^T(A^\Delta+\diag(d^\Delta)) x}{\|x\|^2}\\
  &=\max_{x\in\swozero} \frac{x^T B^\opt x -2x_ux_v-x_u^2-x_v^2}{\|x\|^2}\\
  &=\max_{x\in\swozero} \frac{x^T B^\opt x -(x_u+x_v)^2}{\|x\|^2}\\
  &\leq \max_{x\in\swozero} \frac{x^T B^\opt x}{\|x\|^2}=0.
\end{align*}
For the bisection vector of an optimal bisection size, 
we have $x_u=1,x_v=-1$ or $x_u=-1,x_v=1$ and hence the last inequality is equality. 
We can conclude 
\[g(G',d')=\frac{\tsum(B')-n\lambda(B_S')}{4}=\frac{\tsum(B')-0}{4}=\frac{4\bw(G)-4}{4}=\bw(G)-1.\]
This completes the proof of the theorem.\qed 

\subsection*{Proof of Lemma~\ref{lemma:dopt:unique}}

The following fact is needed for the proof of Lemma~\ref{lemma:dopt:unique}. It has been observed independently in~\cite{Blumofe1993}.
\begin{lemma}
  Let $G$ be a graph with $h(G)=\bw(G)$
  and let $y$ be the bisection vector of an arbitrary optimum solution.
  Then for every $d^\opt$,
  with $g(G,d^\opt)=\bw(G)$ and $\sum_i d^\opt_i=4\bw(G)-2|E|$, there exists some $\alpha^{(y)}\in\mathbb{R}$ such that $d^\text{opt}=d^{(y)}+\alpha^{(y)} y$.
\label{theorem:dy:alpha}
\end{lemma}

\begin{proof}
The assumptions imply $\lambda(B^\opt_S)=0$ with $B^\opt=A+\diag(d^\opt)$. Next,
due to a fact stated in Lemma~\ref{lemma:y:largest:eigenvector}, $y$ is an eigenvector of $B^\opt_S$ corresponding to the largest eigenvalue $0$. We get the following sequence 
of equivalent conditions:
\begin{align*}
&& B^\opt_S y &= 0y=(0,\ldots,0)^T\\
\Leftrightarrow && PB^\opt P y &=(0,\ldots,0)^T\\
\Leftrightarrow && PB^\opt y &=(0,\ldots,0)^T \eqcomment{y \text{ has mean zero}}\\
\intertext{Since $P$ projects into the zero vector only vectors 
of the  subspace spanned by the identity vector, 
thus we can continue for some
$\alpha\in\mathbb{R}$ }
\Leftrightarrow && B^\opt y &=\alpha (1,\ldots,1)^T\\
\Leftrightarrow && (A+D^\opt) y &=\alpha (1,\ldots,1)^T \eqcomment{D^\opt=\diag(d^\opt)}\\
\Leftrightarrow && Ay+D^\opt y &=\alpha (1,\ldots,1)^T\\
\Leftrightarrow && D^\opt y &=-Ay+\alpha (1,\ldots,1)^T\\
\intertext{In the next step, we multiply the vectors in the equation with the diagonal matrix $\diag(y)$. Since the $y_i\in\{1,-1\}$, the multiplication is revertible and hence ``$\Leftrightarrow$''.}
\Leftrightarrow && \diag(y)D^\opt y &=-\diag(y)Ay+\alpha \diag(y)(1,\ldots,1)^T\\
\Leftrightarrow && d^\opt &=-\diag(y)Ay+\alpha y\\
\Leftrightarrow && d^\opt &=d^{(y)}+\alpha y \eqcomment{\text{Def.~\ref{def:dy}}}
\end{align*}
This completes the proof. Note that $\sum_i y_i=0$ and $\sum_i d^{(y)}_i=4\bw(G)-2|E|$.
\end{proof}

\begin{proof}[Proof of Lemma~\ref{lemma:dopt:unique}]
Consider two optimum bisections with bisection vectors $y$ and $y'$. (Note that 
we consider $y$ and $-y$ as same bisection.) 
For contradiction, assume there are two different $d^\opt_1\ne d^\opt_2$
(up to a constant transition). 
Due to Lemma~\ref{theorem:dy:alpha} 
we have that for every  $y$ 
representing an optimum bisection values $d^\opt_1$ and $d^\opt_2$  can be expressed as
$d^\opt_1=d^{(y)}+\alpha_1 y$ and $d^\opt_2=d^{(y)}+\alpha_2 y$. The difference is then
\[d^\opt_1-d^\opt_2=(\alpha_1-\alpha_2) y.\]
For $y'$ representing an optimum bisection, 
we have analogously $d^\opt_1=d^{(y')}+\beta_1 y'$ and $d^\opt_2=d^{(y')}+\beta_2 y'$ with difference
\[d^\opt_1-d^\opt_2=(\beta_1-\beta_2) y'.\]
We conclude
\[(\alpha_1-\alpha_2) y=(\beta_1-\beta_2) y'.\]

Since $y$ and $y'$ are linearly independent, 
we conclude $\alpha_1=\alpha_2$ and $\beta_1=\beta_2$. 
This means, if there are two optimum bisections, then 
there is only one $d^\opt$ and $\alpha$ is unique!
\end{proof}

\subsection*{Proof of Theorem~\ref{thm:main:neg:res}}
Let $G$ be a graph sampled from the subcritical phase and $(V_1,V_{-1})$ be the planted bisection. Coja-Oghlan \cite{Coja2005} defines two sets of vertices:
\[N_i=\{v\in V_i: e(v,V_i)=e(v,V_{-i})\}\]
\[N^*_i=\{v\in N_i: N(v)\setminus \mathrm{core}(G)=\emptyset\}\]
Let further $(Y_1,Y_{-1})$ be an optimal bisection. Coja-Oghlan claims that, w.h.p., $\#(Y_i\cap N^*_i)\geq \mu/8$ (eventually swap the parts), where $\mu=E(\#N_{1}+\#N_{-1})$ and $\mu \geq n^{1-\Theta(\gamma)}$ with $n(p'-p)=\sqrt{np'\cdot \gamma \ln n}$, $\gamma=\caO(1)$ \cite[page  122]{cojahabil}. Then there are $\mathrm{exp}(\Omega(\mu))$ many optimal bisections.
On the other hand, we will show that, assuming that Boppana works on $G$, the probability that $\#(Y_i\cap N^*_i)\geq 2$ will tend to 0, which means that with w.h.p., Boppana will not work on $G$.

Consider any pair of vertices $v_1\in Y_1\cap N^*_1$ and $v_{-1}\in Y_{-1}\cap N^*_{-1}$. $v_1$ and $v_{-1}$ are not connected by an edge, since they have only neighbors in the core of $G$. Furthermore, they both have balanced degree. Thus, we can apply Lemma~\ref{lemma:boppana:same:neighbors} and conclude, that $v_1$ and $v_{-1}$ have the same neighbors. In direct consequence, all vertices in $Y_i\cap N^*_i$, $i\in\{1,-1\}$ have the same neighbors and the same number of edges to each part as well. We denote this number by $k=e(v_1,V_1)$.

In the following, we will consider sets of 4 vertices, while two are chosen from $Y_1\cap N^*_1$ and two from $Y_{-1}\cap N^*_{-1}$. By our assumption of $\#(Y_i\cap N^*_i)\geq 2$, we can choose at least one such set w.h.p.

Let us first rule out two edge cases. In the first case, the vertices have degree $k=0$. Then Boppana does not work due to Theorem~\ref{theorem:no:two:isolated}. In the second case, the vertices have maximal many edges, i.e. $k=n/2-2$ many edges to each part. W.h.p., a graph does not even have two vertices in each part with $k$ edges:
\[\frac{(n/2)^2 (n/2-1)^2}{4} (p'^{n/2-2}p^{n/2-2})^4 (1-p)^4 (1-p')^2\rightarrow 0\]

Thus, we have to consider $1\leq k\leq n/2-2$. Let $C^{(k)}_i=\{v\in V_i : e(v,V_i)=e(v,V_{-i})=k\}$ be the set of vertices with a balanced number of exactly $k$ edges to each part. With the $k$ from above, we have $Y_i\cap N^*_i\subseteq C^{(k)}_i$.

We want to estimate the expected number of 4-element sets $\{v_1,u_1,v_{-1},u_{-1}\}\subseteq C^{(k)}_1 \cup C^{(k)}_{-1}$ with $v_1,u_1\in C^{(k)}_1$ and $v_{-1},u_{-1}\in C^{(k)}_{-1}$, where all vertices have the same neighbors. Let us take $v_1$ as reference vertex and thus the $k$ edges from $v_1$ to $V_1$ as well as $k$ edges to $V_{-1}$ are given. Now we estimate the probability, that $v_{-1},u_1,u_{-1}$ have exactly the same neighbors. For each vertex and each part, the $k$ neighbors are chosen independently, since the four vertices are not connected to each other. In both parts, there are $n/2-2$ possible neighbors. This makes $\binom{n/2-2}{k}\geq \binom{n/2-2}{1}=n/2-2$ possibilities for the $k$ edges in one part and only one of them coincides with the edges of $v_1$. For 3 vertices to have the same neighbors as $v_1$ in two parts each, the probability is at most $\frac{1}{(n/2-2)^6}$. The expected number of 4 vertices as described with the same neighbors is therefore
\[E(\#4-\mathrm{elem-set})\leq \binom{n/2}{2}^2 \cdot \frac{1}{(n/2-2)^6}\leq \frac{(n/2)^4}{(n/2-2)^6}\rightarrow 0\]

This means, w.h.p. we will not find any 4-element set. In consequence, $\#(Y_i\cap N^*_i)\geq 2$ may not be true w.h.p.\qed

\subsection*{Proof of Lemma~\ref{lemma:boppana:same:neighbors}}
  Let $y$ be the bisection vector corresponding to the optimal bisection in the lemma. Let $v_i\in Y_i$, $i\in\{1,-1\}$ be vertices as in the lemma, which fulfill $e(v_i,Y_i)=e(v_i,Y_{-i})$. We obtain the bisection  vector $y'$ as vector corresponding to $(Y_1\setminus \{v_1\} \cup \{v_{-1}\},Y_{-1}\setminus \{v_{-1}\} \cup \{v_1\})$. Due to the balanced degree, this bisection is optimal as well.

Hence, we have two optimal bisections and from Lemma~\ref{lemma:dopt:unique} we know, that the $d^\opt$ is unique and there are unique $\alpha^{(y)}$ and $\alpha^{(y')}$ corresponding to $y$ and $y'$, resp. It holds
\begin{align*}
&  d^{(y)}+\alpha^{(y)}y=d^{(y')}+\alpha^{(y')}y'\\
\Leftrightarrow &  d^{(y)}-d^{(y')}=\alpha^{(y')}y'-\alpha^{(y)}y
\end{align*}
Since $v_1$ has balanced degree and is only connected to vertices, which are in the same part in $y$ and $y'$, we have $d^{(y)}_{v_1}-d^{(y')}_{v_1}=0$. Furthermore, $y_{v_1}=1$, $y'_{v_1}=-1$. Thus we conclude by the equation above, that $-\alpha^{(y')}-\alpha^{(y)}=0$.

Since $y$ and $y'$ are optimal bisections and $e(v_i,Y_i)=e(v_i,Y_{-i})$, we have
\[\sum_{i\in Y_1\setminus \{v_1\}} d^{(y)}_i-\sum_{i\in Y_1\setminus \{v_1\}} d^{(y')}_i=0\]
because
\[
\begin{array}{rcl}
  \sum_{i\in Y_1\setminus \{v_1\}} d^{(y)}_i 
  & = &
  \bw(G) -e(v_1,Y_{-1})- 2 \cdot | (Y_1\setminus \{v_1\})\times (Y_1\setminus \{v_1\}) \cap E(G)|
  - e(v_1,Y_{1}) \\[3mm]
    & = &
  \sum_{i\in Y_1\setminus \{v_1\}} d^{(y')}_i .
\end{array}
\]
But we have also
\[\sum_{i\in Y_1\setminus \{v_1\}} d^{(y)}_i-d^{(y')}_i=(n/2-1)(\alpha^{(y')}-\alpha^{(y)})=-2\alpha^{(y)}(n/2-1)\]
Thus, $\alpha^{(y)}=\alpha^{(y')}=0$. It follows $d_i^{(y)}-d_i^{(y')}=0$, so that each vertex must have no edge to $v_1$ and $v_{-1}$ or must have an edge to both of them. Hence, the $v_1$ and $v_{-1}$ have exactly the same neighbors.\qed

\subsection*{Proof of Theorem~\ref{theorem:no:two:isolated}}
Let $A$ be the adjacency matrix of $G$ and
\[A'=\left(\begin{tabular}{c|cc}A & 0 & 0\\\hline 0&0&0\\0&0&0\end{tabular}\right)\]
be the adjacency matrix of $G'$, where we added two isolated vertices to $G$. Since $h(G)=\bw(G)$,
there exists a $d^\text{opt}$, such that $g(G,d^\text{opt})=\bw(G)$ and $\lambda((A+d^\opt)_S)=0$. It then holds $\sum_i d^\opt_i=4\bw(G)-2|E|$.
\begin{align*}
& h(G)-h(G')\\ 
=&g(G,d^\opt)-\max_{d'} g(G',d')\\
=& \frac{\tsum(A)+\tsum(d^\opt)}{4}-\max_{d'}\frac{\tsum(A')+\tsum(d')-(n+2)\lambda(B_S')}{4} \eqcomment{B'_S=(A'+\diag(d'))_S}\\
=& \frac{\tsum(d^\opt)}{4}-\max_{d'}\frac{\tsum(d')-(n+2)\lambda(B_S')}{4} \eqcomment{\tsum(A)=\tsum(A')}\\
=& \frac{\tsum(d^\opt)}{4}-\max_z\frac{\tsum(z+(d_\opt^T,0,0)^T)-(n+2)\lambda(B_S')}{4} \eqcomment{d'=z+(d_\opt^T,0,0)^T}\\
=&-\max_z\frac{\tsum(z)-(n+2)\lambda(B_S')}{4}\\
=&\min_z\left(\frac{n+2}{4}\lambda(B_S')-\frac{\tsum(z)}{4}\right)\\
=&\min_z\left(\frac{n+2}{4}\max_{x\in\swozero} \frac{x^T (A'+\diag (d')) x}{\|x\|^2}-\frac{\tsum(z)}{4}\right)\\
=&\min_z\left(\frac{n+2}{4}\max_{x\in\swozero} \left(\frac{x^T (A'+\diag(z+(d_\opt^T,0,0)^T)) x}{\|x\|^2}\right)-\frac{\tsum(z)}{4}\right)
\intertext{We restrict ourselves two two kinds of vector $x_a=(x_1,\ldots,x_n,0,0)^T$ with $\sum_{i=1}^n x_i=0$ and $x_b=(1,\ldots,1,-\frac{n}{2},-\frac{n}{2})$:}
\geq &\min_z\left(\frac{n+2}{4}\max_{x\in \{x_a,x_b\}} \left(\frac{x^T (A'+\diag(z+(d_\opt^T,0,0)^T)) x}{\|x\|^2}\right)-\frac{\tsum(z)}{4}\right)
\eqlabel{isolated:drop:eqn}
\end{align*}
We want to show that this term is at least $\frac{4\bw(G)}{n^2}$. Therefore, we analyze the $\max$-term separately and then show, for which $d'$ we have to choose which of the $x_a$ and $x_b$.

Firstly, consider vector $x_a$. Let $z^{(n)}$ denote the first $n$ components of vector $z$. Then
\begin{align*}
  &\max_{x_a=(x_1,\ldots,x_n,0,0)^T, \sum_{i=1}^n x_i=0} \left(\frac{x_a^T (A'+\diag(z+(d_\opt^T,0,0)^T)) x_a}{\|x_a\|^2}\right)\\
 =&\max_{\sum_{i=1}^n x_i=0} \left(\frac{x^T (A+\diag(d^\opt)) x}{\|x\|^2} + \frac{x^T \diag(z^{(n)}) x}{\|x\|^2}\right)\\
 =&\max_{\sum_{i=1}^n x_i=0} \left(\frac{x^T B x}{\|x\|^2} + \frac{x^T \diag(z^{(n)}) x}{\|x\|^2}\right)
\intertext{We choose an optimal bisection vector $y$ of $G$:}
\geq&\frac{y^T B y}{\|y\|^2} + \frac{y^T \diag(z^{(n)}) y}{\|y\|^2}=\frac{\sum_{i=1}^n z_i}{n} \eqlabel{tmp:add:two:vec:xa} \eqcomment{\text{Lemma~\ref{lemma:y:largest:eigenvector}}}\\  
\end{align*}

Secondly, we consider $x_b=(1,\ldots,1,-\frac{n}{2},-\frac{n}{2})^T$:
\begin{align*}
 & \frac{x_b^T (A'+\diag(z+(d_\opt^T,0,0)^T)) x_b}{\|x_b\|^2}\\
=& \frac{\tsum(A)+\sum_i d_i^\opt+x_b^T \diag(z) x_b}{\|x_b\|^2}\\
=& \frac{4\bw(G)+x_b^T \diag(z) x_b}{\|x_b\|^2} \eqcomment{\sum_i d^\opt_i=4\bw(G)-2|E|}\\
=& \frac{4\bw(G)}{(n+2)\left(\frac{n}{2}\right)} + \frac{\sum_{i=1}^n z_i+(z_{n+1}+z_{n+2})\left(\frac{n}{2}\right)^2}{(n+2)\left(\frac{n}{2}\right)} \eqlabel{tmp:add:two:vec:xb}
\end{align*}

We insert the result \eqref{tmp:add:two:vec:xa} for $x_a$ and \eqref{tmp:add:two:vec:xb} for $x_b$ into \eqref{isolated:drop:eqn}:
\[\eqref{isolated:drop:eqn}\geq \min_z\left(\frac{n+2}{4}\max \left(
  \frac{\sum_{i=1}^n z_i}{n},
  \frac{4\bw(G)}{(n+2)\left(\frac{n}{2}\right)} + \frac{\sum_{i=1}^n z_i+(z_{n+1}+z_{n+2})\left(\frac{n}{2}\right)^2}{(n+2)\left(\frac{n}{2}\right)}\right)-\frac{\tsum(z)}{4}
\right)\]

We again simplify the terms separately for \eqref{tmp:add:two:vec:xa}
\begin{align*}
&\frac{n+2}{4} \frac{\sum_{i=1}^n z_i}{n} - \frac{\sum_{i=1}^{n+2} z_i}{4}\\
=&\frac{(n+2)\sum_{i=1}^n z_i-n\sum_{i=1}^n z_i-n(z_{n+1}+z_{n+2})}{4n}\\
=&\frac{2\sum_{i=1}^n z_i-n(z_{n+1}+z_{n+2})}{4n}\\
=&\frac{\sum_{i=1}^n z_i}{2n}-\frac{z_{n+1}+z_{n+2}}{4}=\frac{1}{2}\delta \eqcomment{\delta=\frac{\sum_{i=1}^n z_i}{n}-\frac{z_{n+1}+z_{n+2}}{2}}
\intertext{and \eqref{tmp:add:two:vec:xb}}
 & \frac{n+2}{4} \eqref{tmp:add:two:vec:xb}-\frac{\sum_{i=1}^{n+2} z_i}{4}\\
=& \frac{4\bw(G)}{2n} + \frac{\sum_{i=1}^n z_i+(z_{n+1}+z_{n+2})\left(\frac{n}{2}\right)^2}{2n}-\frac{\sum_{i=1}^n z_i}{4}-\frac{z_{n+1}+z_{n+2}}{4}\\
=& \frac{4\bw(G)}{2n} + \left(\frac{1}{2n}-\frac{1}{4}\right)\sum_{i=1}^n z_i+\left(\frac{n}{8}-\frac{1}{4}\right)(z_{n+1}+z_{n+2})\\
=& \frac{4\bw(G)}{2n} + \frac{2-n}{4n}\sum_{i=1}^n z_i+\frac{n-2}{8}(z_{n+1}+z_{n+2})\\
=& \frac{4\bw(G)}{2n} + \frac{2-n}{4}\left(\frac{\sum_{i=1}^n z_i}{n}-\frac{z_{n+1}+z_{n+2}}{2}\right)\\
=& \frac{2\bw(G)}{n} + \frac{2-n}{4}\delta.
\end{align*}
In both cases, the minimization over $z$ could be reduced to a minimization over $\delta$ and we conclude
\[h(G)-h(G')\geq \eqref{isolated:drop:eqn}\geq \min_\delta \max \left(\frac{1}{2}\delta, \frac{2b}{n}+\frac{2-n}{4}\delta\right).\]
The first term in the maximum is monotone increasing and the second one monotone decreasing (for $n\geq 3$). Hence, the minimum is at the intersection point of these two lines:
\begin{align*}
\frac{1}{2}\delta_\mathrm{min}&=\frac{2\bw(G)}{n}+\frac{2-n}{4}\delta_\mathrm{min}\\
\frac{2-2+n}{4}\delta_\mathrm{min}&=\frac{2\bw(G)}{n}\\
\frac{n}{4}\delta_\mathrm{min}&=\frac{2\bw(G)}{n}\\
\delta_\mathrm{min}&=\frac{8\bw(G)}{n^2}
\end{align*}
It follows
\[h(G)-h(G')\geq \frac{1}{2} \delta_\mathrm{min}=\frac{4\bw(G)}{n^2}.\]\qed

\subsection*{Proof of Proposition~\ref{prop:Bopp:reformulation}}

To obtain an SDP formulation 
we start with Boppana's function $h(G)$ and transform it 
successively as follows:
\begin{align*}
h(G)=&\max_{d\in\mathbb{R}^n} \frac{\tsum(A+\diag(d))-n\lambda((A+\diag(d))_S)}{4}\\
=&\max_{d\in\mathbb{R}^n} \frac{J\bullet A+\myvec{1}^T d-n\lambda(P(A+\diag(d))P)}{4}\\
=&\frac{J\bullet A}{4}+\frac{1}{4}\max_{d\in\mathbb{R}^n} \left(\myvec{1}^T d-n\lambda(P(A+\diag(d))P)\right)\\
=&\frac{J\bullet A}{4}+\frac{1}{4}\max_{d\in\mathbb{R}^n} \left(-n\lambda(P(A+\diag(d))P-\frac{\myvec{1}^T d}{n}I)\right)\\
=&\frac{J\bullet A}{4}-\frac{n}{4}\min_{d\in\mathbb{R}^n} \lambda\left(P(A+\diag(d))P-\frac{\myvec{1}^T d}{n}I\right)\\
=&\frac{J\bullet A}{4}-\frac{n}{4}\min_{d\in\mathbb{R}^n} \lambda(M(d)),
\end{align*}
where $M(d)=P(A+\diag(d))P-\frac{\myvec{1}^T d}{n}I$.
Hence, we want to solve the following problem: Minimize the largest eigenvalue of the matrix $M(d)$ for $d\in\mathbb{R}^n$. For this problem, \cite{vandenberghe1996} gives the SDP formulation:
\[\min z\quad\text{ s.t. }\quad zI-M(d)\succeq 0,\]
with $z\in\mathbb{R}, d\in\mathbb{R}^n$.
Inserting $M(d)$ and then substituting $z$ with $z-\frac{\myvec{1}^T d}{n}$, we get
\[
   \min_{z\in\mathbb{R}, d\in\mathbb{R}^n} 
   \left(z-\frac{\myvec{1}^T d}{n}\right)\quad\text{ s.t. }\quad zI - P(A+\diag(d))P\succeq 0.\]
It is easy to see that the constraint matrix above is equal to the
constraint matrix of~\eqref{primal:sdp:boppana}, since $P=I-\frac{J}{n}$.
This completes the proof that $h(G)$ maximized by Boppana's algorithm 
gives the same value as the optimum solution of~\eqref{primal:sdp:boppana}
because under the constraints we have
\[
  h(G)=\frac{J\bullet A}{4}-\frac{1}{4}\min_{z\in\mathbb{R}, d\in\mathbb{R}^n} 
 (nz-\myvec{1}^T d) .
 \]

To obtain the formulation for a dual program, consider the primal SDP in the form:
\[
  \min_{z\in\mathbb{R}, d\in\mathbb{R}^n} 
  (nz-\myvec{1}^T d) \quad\text{ s.t. }\quad  -PAP+zI-\sum_i d_i PI_i P \ \succeq \ 0,  
\]
where $I_i$ denotes the matrix which has a single 1
in the $i$th row and the $i$th column and zero everywhere else.
The dual can be derived by using  
the rules~\eqref{eq:def:sdp:dual}. We obtain:
\[
  \max_{Y\in\mathbb{R}^{n\times n}} (PAP)\bullet Y\quad{s.t.}\quad
    I\bullet Y =n,\quad
   \forall i:  \  -PI_i P\bullet Y=-1,\quad
  Y\succeq 0.
\]
Thus, since  $P=I-\frac{J}{n}$, we get the following formulation for the dual SDP:
\[
  \max_{Y\in\mathbb{R}^{n\times n}} \left(A-\frac{JA+AJ}{n}+\frac{\tsum(A)J}{n^2}\right)\bullet Y
\]
under the constraints:
\begin{align*}
  \sum_i y_{ii}&=n,\\
  \forall i  \quad y_{ii}-\frac{1}{n}\sum_j y_{ji}-\frac{1}{n} \sum_j y_{ij}+\frac{1}{n^2} \sum_{k,j} y_{kj}&=1,\\
  Y&\succeq 0.
\end{align*}

Here we can note that the second constraint is equal to $(P Y P)_{ii}=1$,
for all $i$. Note further that $(AJ)\bullet Y=\sum_i \deg(i) \sum_j y_{ij}$
and an analogous holds for $(JA)\bullet Y$. Hence, we can reformulate the objective function as follows:
\[
   \max_{Y\in\mathbb{R}^{n\times n}}  \left(A\bullet Y-\frac{1}{n}\sum_j \deg(j) \sum_i y_{ij}-\frac{1}{n}
   \sum_i \deg(i) \sum_j y_{ij}+\frac{1}  {n^2}\sum_{i,j}y_{ij}\right).
\] 
This completes the proof.\qed

\subsection*{Proof of Theorem~\ref{theorem:feige:is:dual:boppana}}
We start with the following fact:
\begin{claim}
  Let $X$ be a positive semidefinite matrix. Then the conditions 
  $(a)$ $\forall i: \sum_j x_{ij}=0$ and 
  $(b)$ $\sum_{i,j} x_{ij}=0$ are equivalent.
\label{lemma:semidefinite:row:sum}
\end{claim}

\begin{proof}
  We show two directions. If $(a)$ holds, it follows directly that $(b)$ is true as well. 
  We proceed with proving of the second direction and assume, that $(b)$ holds.

Each positive semidefinite matrix~$X$ 
can be represented as a Gram matrix,  
i.e. as matrix of scalar products $x_{ij}=\langle u_i,u_j\rangle$ of vectors~$u_i$.
  Thus, we have
\[\sum_{i,j} x_{ij}=\sum_{i,j} \langle u_i,u_j\rangle=\sum_i \langle u_i,\sum_j u_j\rangle=\langle \sum_i u_i,\sum_j u_j\rangle=0,\]
where we used condition~$(b)$.
The scalar product of the vector $\sum_i u_i$ with itself is zero and we conclude that it is the zero vector: $\sum_i u_i=\myvec{0}$. Now we compute
\[\sum_j x_{ij}=\sum_j \langle u_i,u_j\rangle=\langle u_i,\sum_j u_j\rangle=\langle u_i,\myvec{0}\rangle=0\]
which gives condition~$(a)$. 
\end{proof}

Now we ready to prove Theorem~\ref{theorem:feige:is:dual:boppana}.
For convenience we restate the primal SDP~\eqref{primal:sdp:feige:kilian} 
as follows: 
\begin{equation}\label{n:primal:sdp:feige:kilian}
  h_p(G) = \min_Y \left(\frac{|E|}{2} -\frac{1}{4}(A\bullet Y)\right) \quad\text{ s.t. }\quad
  \ \forall i \ y_{ii}=1, \ \sum_{i,j} y_{ij} =0,  \text{ and } Y\succeq 0,
\end{equation}
We show that for the following program
\begin{equation} \label{eq:tmp:obj.f}  
  h'_p(G) = \max_Y A\bullet Y \\
\end{equation}
under the constraints:
\begin{align*}
  \forall i: y_{ii}&=1,\\
  \sum_{i,j} y_{ij}&=0,\\
  Y&\succeq 0,
\end{align*}
we have $h'_p(G)=d(G)$, where recall, $d(G)$ is the objective function of~\eqref{dual:sdp:boppana}.
Then we conclude $h_p(G)=\frac{|E|}{2} -\frac{1}{4}h'_p(G)=\frac{|E|}{2} -\frac{1}{4}d(G)$.

Consider an optimal solution matrix~$Y$ for the SDP. 
We show that $Y$ is a solution to the dual program~\eqref{dual:sdp:boppana} as well,
with the value for the objective function equal to~\eqref{eq:tmp:obj.f}.
 
Since $y_{ii}=1$, the first constraint  of ~\eqref{dual:sdp:boppana} is fulfilled. 
Due to Claim~\ref{lemma:semidefinite:row:sum} and since $\sum_{i,j} y_{ij}=0$, we have $\sum_j y_{ij}=0$ for all $i$. Hence, the second constraint of~\eqref{dual:sdp:boppana}:
\[
  \forall i\quad y_{ii}-\frac{1}{n}\sum_j y_{ji}-\frac{1}{n} \sum_j y_{ij}+\frac{1}{n^2} \sum_{k,j} y_{kj}=1
\]
is fulfilled as well. In the objective function of~\eqref{dual:sdp:boppana}, 
the second and third term are zero, since $(AJ)\bullet Y=\sum_i \deg(i) \sum_j y_{ij}=0$.
Obviously, the fourth term is zero due to the constraints as well.
Hence, we obtain the same value as~$h'_p(G)$.

For the other direction, consider an optimum 
solution matrix~$Y$ of SDP~\eqref{dual:sdp:boppana}. 
First we show that the first and second constraint
of~\eqref{dual:sdp:boppana}  imply $\sum_{i,j} y_{ij}=0$:
\begin{align*}
&& 
\forall i \quad  y_{ii}-\frac{1}{n}\sum_j y_{ji}-\frac{1}{n} \sum_j y_{ij}+\frac{1}{n^2} 
  \sum_{k,j} y_{kj}&=1 \eqcomment{\text{second contraint of~\eqref{dual:sdp:boppana} for each $i$}}\\
\Rightarrow && \sum_i y_{ii}-\frac{1}{n}\sum_{i,j} y_{ji}-\frac{1}{n} \sum_{i,j} y_{ij}+\frac{n}{n^2} \sum_{i,j} y_{ij}&=n \eqcomment{\text{sum all $n$ constraints}}\\
\Rightarrow && n-\frac{1}{n}\sum_{i,j} y_{ji}-\frac{1}{n} \sum_{i,j} y_{ij}+\frac{n}{n^2} \sum_{i,j} y_{ij}&=n \eqcomment{\text{use the  first contraint of~\eqref{dual:sdp:boppana}.}}\\
\Rightarrow && \sum_{i,j} y_{ij}&=0.
\end{align*}
Next, due to Claim~\ref{lemma:semidefinite:row:sum} we know
that  $\sum_j y_{ij}=0$ for all $i$. Again from the  second constraint, 
of~\eqref{dual:sdp:boppana}
we conclude that $y_{ii}=1$. Hence, the constraints 
of the SDP~\eqref{primal:sdp:feige:kilian} are fulfilled. 
Obviously, the second, third and fourth term in the objective function 
of~\eqref{dual:sdp:boppana}
are zero again and the objective values of both SDPs are the same as well.\qed

\subsection*{Lemma~\ref{lemma:necessary:many:edges} and its Proof}

\begin{lemma}[Necessary for many edges]
  Let $G=(V,E)$ be a graph and $y$ an optimal bisection vector of $G$. For $i\in\{+1,-1\}$ let $C_i=\{u \mid y_u=i \wedge \exists v: y_v=-i \wedge \{u,v\}\in E\}$ be the set of vertices in part $i$ located at the cut. If there exist non-empty $\tilde C_i\subseteq C_i$ with $k=\min\{|\tilde C_{+1}|,|\tilde C_{-1}|\}$, $k+\delta=\max\{|\tilde C_{+1}|,|\tilde C_{-1}|\}$, $l=|V|-(k+\delta)$, s.t.
  \begin{compactitem}
    \item $(3k<l \,\land\, \delta=0)$ or $(4k<l \,\land\, \delta<\min \{\frac{4k^2}{l-4k},\frac{7}{128}l\})$
    \item $2 |E(\tilde C_{+1},\tilde C_{-1})| \geq |E(\tilde C_{+1}\cup \tilde C_{-1},V\setminus (\tilde C_{+1}\cup \tilde C_{-1}))|$,
  \end{compactitem}
then $h(G)<\bw(G)$.
\label{lemma:necessary:many:edges}
\end{lemma}

An illustration is found in Figure~\ref{fig:necessary:many:edges}. The parameter $\delta$ allows for some unbalaced size of the subsets.

\begin{figure}[h!]
\begin{center}
  \begin{tikzpicture}[scale=0.8, every node/.style={scale=0.8}]
    \draw (0,0) ellipse (0.5 and 1.5);
    \node at (0,-0.5) {$\tilde C_{+1}$};
    \node at (-0.6,1.2) {$\delta$};
    \node at (-0.6,-1.2) {$k$};
    \draw[dashed] (-0.45,0.7) -- (0.45,0.7);
    \draw[dashed] (1,1.5) -- (1,-2);

    \draw (2,-0.5) ellipse (0.5 and 1);
    \node at (2,-0.5) {$\tilde C_{-1}$};
    \node at (2.6,-1.2) {$k$};

    \draw (0,0) -- (2,0);
    \draw (0,1.2) -- (2,0);
    \draw (0,-1.2) -- (2,-1.2);
    \draw (2,0) -- (3,1);
    \draw (2,-1.2) -- (3,-1.7);
    \draw (2,-1.2) -- (0,-1.7);
  \end{tikzpicture}
\caption{Forbidden graph structures as in Lemma~\ref{lemma:necessary:many:edges}.}
\label{fig:necessary:many:edges}
\end{center}
\end{figure}
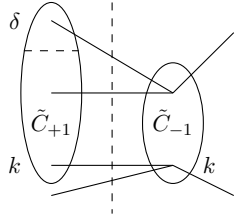

\begin{proof}
  For contradiction, we assume $h(G)=\bw(G)$. For the bisection vector $y$, we then have $d^\opt=d^{(y)}+\alpha y$ for some $\alpha\in\mathbb{R}$. Then $\lambda(B_S)=0$ for $B=A+\diag(d^\opt)$. We will contradict this by choosing a vector $x$ and then show that the Rayleigh quotient for $x$ and $B_S$ is larger than 0 (for any $\alpha$). W.\,l.\,o.\,g. we assume $|\tilde C_{+1}|\geq |\tilde C_{-1}|$. We choose

\[x_i=\begin{cases}
  -1 & \text{ if } y_i=+1 \wedge i\not\in \tilde C_{+1},\\
  z & \text{ if } i \in \tilde C_{+1}\cup \tilde C_{-1},\\
  -\beta z & \text{ if } y_i=-1 \wedge i\not\in \tilde C_{-1},\\
\end{cases}\]
with $\beta=\sqrt{\frac{\delta+l/z^2}{\delta+l}}$ and $z=\frac{2kl+\delta l+2\sqrt{kl(k +\delta)(l+\delta)}}{4k^2+4\delta k-\delta l}$. Note that for $\delta=0$, we have $z=l/k>3$, $\beta=1/z<1/3$ and $-\beta z=-1$.

First we derive the $z$ above by enforcing $\sum_i x_i=0$ and choosing $\beta$ as above:
\begin{align*}
\sum_i x_i&=l (-1) + (k+\delta) z + k z + (\delta+l)(-\beta z)\\
&=-l + (k+\delta) z + k z - (\delta+l)\sqrt{\frac{\delta z^2+l}{\delta+l}}\\
&=-l + (2k+\delta) z - \sqrt{(\delta+l)(\delta z^2+l)}\overset{!}{=}0
\end{align*}
\begin{align*}
\Leftrightarrow && \sqrt{(\delta+l)(\delta z^2+l)}&=(2k+\delta) z -l\\
\Rightarrow && (\delta+l)(\delta z^2+l)&=((2k+\delta) z -l)^2\\
\Leftrightarrow && \delta^2 z^2+\delta l +\delta l z^2+l^2&=(2k+\delta)^2 z^2 +l^2 -2(2k+\delta)lz\\
\Leftrightarrow && \delta^2 z^2+\delta l +\delta l z^2&=4k^2 z^2 +\delta^2 z^2 +4k\delta z^2 -4klz-2\delta lz\\
\Leftrightarrow && 0&=(4k^2 +4k\delta - \delta l) z^2 +(-4kl-2\delta l)z - \delta l\\
\Leftrightarrow && z&=\frac{2kl+\delta l\pm \sqrt{(2kl+\delta l)^2+\delta l (4k^2 +4k\delta - \delta l)}}{4k^2 +4k\delta - \delta l}\\
\Leftrightarrow && z&=\frac{2kl+\delta l\pm \sqrt{4(kl)^2+4kl\delta l+\delta l (4k^2 +4k\delta)}}{4k^2 +4k\delta - \delta l}\\
\Leftrightarrow && z&=\frac{2kl+\delta l\pm 2\sqrt{kl(kl+\delta l+\delta (k +\delta))}}{4k^2 +4k\delta - \delta l}\\
\Leftrightarrow && z&=\frac{2kl+\delta l\pm 2\sqrt{kl(k +\delta)(l+\delta)}}{4k^2 +4k\delta - \delta l}\\
\end{align*}
We take the larger $z$-solution with the $+$.

We show that by our choice of $\beta$, the sum of squares for both parts is the same:
\begin{align*}
  \sum_{i: y_i=+1} x_i^2 - \sum_{i: y_i=-1} x_i^2&=(l (-1)^2 + (k+\delta) z^2) - (k z^2 + (\delta+l)(-\beta z)^2)\\
&=l + (k+\delta) z^2 - k z^2 - (\delta+l)\frac{\delta z^2+l}{\delta+l}\\
&=l + (k+\delta) z^2 - k z^2 - (\delta z^2+l)=0
\end{align*}

Thus, $\alpha$ will have no effect:
\begin{align*}
  \frac{x^T B_S x}{\|x\|^2}&=\frac{x^T B x}{\|x\|^2}\\
&=\frac{x^T (A+\diag(d^{(y)}+\alpha y) x}{\|x\|^2}\\
&=\frac{x^T (A+\diag(d^{(y)})) x+x^T (\alpha y) x}{\|x\|^2} \eqcomment{x^T (\alpha y) x=\alpha \sum_i y_i x_i^2=0}\\
&=\frac{x^T (A+\diag(d^{(y)})) x}{\|x\|^2}\eqlabel{tmp:combining:lemma}
\end{align*}

From now we consider the case $4k<l$ and $\delta<\min \{\frac{4k^2}{l-4k},\frac{7}{128}l\}$.
Next we show $z>4$. From $\delta<\frac{4k^2}{l-4k}$, we get for the denominator of $z$ that $4k^2 +4k\delta - \delta l>0$. For the enumerator, we have:
\begin{align*}
  2kl+\delta l+ 2\sqrt{kl(k +\delta)(l+\delta)} &= 2kl+5\delta l+ 2\sqrt{kl(k +\delta)(l+\delta)}-4\delta l\\
  &> 8k^2+20\delta k+ 4\sqrt{k^2(k +\delta)(4k+\delta)}-4\delta l \eqcomment{\text{Assumption } 4k<l}\\
  &= 8k(k+\delta)+12\delta k+ 4\sqrt{k^2 4k^2}-4\delta l\\
  &>16k(k+\delta)-4\delta l\\
  &=4(4k(k+\delta)-\delta l) \eqcomment{\text{4 times denominator of } z}
\end{align*}
Since the enumerator is more than 4 times larger then the denominator and both are positive, we conclude $z>4$. From $\delta<\frac{7}{128}l$ follows further, that $\beta<1/3$:
\begin{align*}
&& \beta^2=\frac{\delta+l/z^2}{\delta+l} & \leq \frac{1}{9}=\left(\frac{1}{3}\right)^2\\
\Leftrightarrow && 9(\delta+l/z^2) &\leq \delta+l\\
\Leftrightarrow && 8\delta &\leq l-\frac{9l}{z^2}\\
\Leftarrow && 8\delta &\leq l-\frac{9l}{16} \eqcomment{z>4}\\
\Leftrightarrow && \delta &\leq \frac{7}{16\cdot 8}l\\
\end{align*}

Now we want to show that \eqref{tmp:combining:lemma} is larger than zero. For this we decompose $B=A+\diag(d^{(y)}$ into $B=\sum_{e\in E} B^e$ and analyze $x^T B^e x$ for each edge $e$ separately. Note that $d^{(y)}_i$ is for vertex $i$ the number of neighbors in the other part minus the number of neighbors in the same part. For the decomposition, we set $B^e_{ii}=B^e_{jj}=1$, if $e=\{i,j\}$ is a cut edge and $B^e_{ii}=B^e_{jj}=-1$, if $e$ is a inner edge. Further, $B_{ij}=B_{ji}=1$.

If $e=\{i,j\}$ is a cut edge, we have $x^T B^e x=2x_ix_j+x_i^2+x_j^2=(x_i+x_j)^2$. Thus, cut edges always contribute positive. We only consider the edges $E(\tilde C_{+1},\tilde C_{-1})$. Since $x_i=x_j=z$, they contribute $4z^2$ each.

If $e=\{i,j\}$ is a inner edge, we have $x^T B^e x=2x_ix_j-x_i^2-x_j^2$. For inner edges in $V\setminus (\tilde C_{+1}\cup \tilde C_{-1})$, $x_i=x_j$ and the contribution is 0. The same holds for inner edges in $\tilde C_{+1}$ and $\tilde C_{-1}$. Thus, we only have to consider the edges $E(\tilde C_{+1}\cup \tilde C_{-1},V\setminus (\tilde C_{+1}\cup \tilde C_{-1}))$. One vertex is $z$, the other $-1$ or $-\beta z<-1$. Thus, the contribution is $-2z-1-z^2$ or $-2\beta z^2-\beta^2 z^2-z^2=-(3\beta+1)z^2$. Since $0<\beta<1/3$, both are larger than $-2z^2$.

We conclude:
\[x^TBx > |E(\tilde C_{+1},\tilde C_{-1})|\cdot 4z^2 + |E(\tilde C_{+1}\cup \tilde C_{-1},V\setminus (\tilde C_{+1}\cup \tilde C_{-1}))| \cdot (-2z^2)\]
By the assumption in the Lemma, this is greater or equal to zero. 
\end{proof}

\end{document}